\documentclass{article}

\usepackage[hidelinks]{hyperref}

\usepackage{geometry}
\usepackage{amsthm}
\usepackage{amsmath}
\usepackage{amssymb}
\usepackage{appendix}
\usepackage{bbold}

\usepackage{tikz}
\usepackage{pgfplots}
\usepackage{pgf}
\usetikzlibrary{patterns}
\usepackage{float}
\usepackage{graphicx}

\usepackage{mathrsfs}

\usepackage{algorithm}
\floatname{algorithm}{Protocol}
\usepackage{algcompatible}

\usepackage{ulem}
 \theoremstyle{plain}
  \newtheorem{thm}{Theorem}
  
  \newtheorem{lemma}[thm]{Lemma}
  \newtheorem{cor}[thm]{Corollary}
  
  \newtheorem*{thm*}{Theorem}
  \newtheorem*{prop*}{Proposition}
  \newtheorem*{lemma*}{Lemma}
  \newtheorem*{cor*}{Corollary}
  \newtheorem*{remark*}{Remark}

\theoremstyle{definition}
\newtheorem{defn}[thm]{Definition}
 \newtheorem*{conj*}{Conjecture}

\newtheorem{innercustomlem}{Lemma}
\newenvironment{customlemma}[1]
  {\renewcommand\theinnercustomlem{#1}\innercustomlem}
  {\endinnercustomlem}

\newtheorem{innercustomthm}{Theorem}
\newenvironment{customthm}[1]
  {\renewcommand\theinnercustomthm{#1}\innercustomthm}
  {\endinnercustomthm}

\newcommand{\Tr}{\mathrm{Tr}}

\newcommand{\sket}[1]{{\ensuremath{\lvert#1\rangle}}}
\newcommand{\lket}[1]{{\ensuremath{\left\lvert#1\right\rangle}}}
\newcommand{\ket}[1]{\mathchoice{\lket{#1}}{\sket{#1}}{\sket{#1}}{\sket{#1}}}

\newcommand{\sbra}[1]{{\ensuremath{\langle#1\rvert}}}
\newcommand{\lbra}[1]{{\ensuremath{\left\langle#1\right\rvert}}}
\newcommand{\bra}[1]{\mathchoice{\lbra{#1}}{\sbra{#1}}{\sbra{#1}}{\sbra{#1}}}

\newcommand{\sketbra}[2]{{\ensuremath{\lvert #1\rangle\langle #2\rvert}}}
\newcommand{\lketbra}[2]{{\ensuremath{\left\lvert #1\middle\rangle\middle\langle #2\right\rvert}}}
\newcommand{\ketbra}[2]{\mathchoice{\lketbra{#1}{#2}}{\sketbra{#1}{#2}}{\sketbra{#1}{#2}}{\sketbra{#1}{#2}}}

\usepackage{todonotes}

\usepackage{authblk}

\begin{document}

\title{Device-independent Certification of\\
One-shot Distillable Entanglement}

\author[1]{Rotem Arnon-Friedman}
\author[2]{Jean-Daniel Bancal}
\affil[1]{Institute for Theoretical Physics, ETH-Z\"urich, CH-8093, Z\"urich, Switzerland}
\affil[2]{Department of Physics, University of Basel, CH-4056, Basel, Switzerland}

\date{\empty}

\maketitle

\begin{abstract}
	Entanglement sources that produce many entangled states act as a main component in applications exploiting quantum physics such as quantum communication and cryptography. 
	Realistic sources are inherently noisy, cannot run for an infinitely long time, and do not necessarily behave in an independent and identically distributed manner. 
	An important question then arises -- how can one test, or certify, that a realistic source produces high amounts of entanglement? 
	Crucially, a meaningful and operational solution should allow us to certify the entanglement which is available for further applications after performing the test itself (in contrast to assuming the availability of an additional source which can produce more entangled states, identical to those which were tested).
	To answer the above question and lower bound the amount of entanglement produced by an uncharacterised source, we present a protocol that can be run by  interacting classically with uncharacterised (but not entangled to one another) measurement devices used to measure the states produced by the source. A successful run of the protocol implies that the remaining quantum state has high amounts of one-shot distillable entanglement. That is, one can distill many maximally entangled states out of the single remaining state. Importantly, our protocol can tolerate noise and, thus, certify entanglement produced by realistic sources.
	With the above properties, the protocol acts as the first ``operational device-independent entanglement certification protocol'' and allows one to test and benchmark uncharacterised entanglement sources which may  be otherwise incomparable. 
\end{abstract}

\section{Introduction}\label{sec:intro}

	Entanglement is one of the most fundamental concepts of quantum physics, distinguishing it from classical physics~\cite{horodecki2009quantum}. Furthermore, it plays a crucial role in the advantages gained by considering applications of quantum physics such as quantum computation~\cite{williams2010explorations}, communication~\cite{de2002quantum}, and cryptography~\cite{broadbent2016quantum}.
	
	For most applications utilising entanglement, a single entangled pair of particles, e.g., a maximally entangled state $\ket{\Phi^+} = \frac{1}{\sqrt{2}}\left(\ket{00}+\ket{11}\right)$, is not sufficient. Instead, one must use many copies of entangled states or, to put differently, a highly entangled state, such as the maximally entangled state  $\ket{\Phi^L}= \frac{1}{\sqrt{L}}\sum_{i=1}^L \ket{i}\ket{i}$ of rank~$L$. Sources that produce high amounts of entanglement are, thus, a prerequisite for gaining from the computational and cryptographic advantages that quantum physics and quantum information unveil. 
	
	Three interesting questions then arise. Firstly, how should one \emph{quantify} the amount of entanglement produced by a source? Secondly, how can one \emph{compare}, or \emph{benchmark}, different  types of entanglement sources? Thirdly, how can one test, or \emph{certify}, that high amounts of entanglement are indeed being produced by the source? 
	
	Importantly, we would like to answer these questions in an operational way. That is: (a)~the suggested answer should be relevant for realistic sources and experimental settings. (b)~When considering possible certification procedures, the statement should apply to the entanglement which is still present after the test rather than the entanglement which was already used and, hence, destroyed by the performed certification procedure itself. 

	The current work is concerned with certifying entanglement in such an operational way. 
	Realistic sources are inherently noisy: they produce at best noisy entangled states. Moreover, a source creating  many entangled pairs might produce systems that are correlated with one another due to, e.g., drifting of the source with time or the presence of a memory inside the source. In other words, the overall state produced by the source (after using it many times) is not an independent and identically distributed (IID) state. We call such sources \emph{noisy non-IID sources}.

	Clearly, in order to make any estimation of the quality of the source one must collect some data regarding the produced entanglement by performing certain measurements on the quantum states produced by the source. Similarly to the considerations regarding the source, the measurement devices might also behave in a noisy and non-IID manner. 

	\subsection{Device-independent entanglement certification (DIEC)}

		We fix as a target to answer our questions while accommodating the noisy non-IID nature of the problem in the so called \emph{device-independent} (DI) framework~\cite{scarani2012device}. In the DI approach one treats the quantum apparatuses as black boxes with which we can only interact classically. That is, we assume no prior knowledge regarding the internal behaviour of the source and the measurement devices. Assumptions are made, however, on the relation that certain devices can have with each other, and on the communication allowed between them. Concretely, we only interact with those devices by ``pushing buttons'' and collecting the classical data output by them. For example, we may push a button on the source apparatus to produce a state and then another button on the measurement device to ask it to perform a measurement with ``input'' $0$ or $1$. The state produced by the source is uncharacterised and we do not know which measurements are actually being performed when we use the inputs $0$ and $1$.

		As widely known, the only way to demonstrate that the actions of physical devices cannot be explained by classical physics in a DI manner is to perform some Bell tests using the devices and observe a violation of the considered Bell inequality~\cite{bell1964,brunner2014bell}. A Bell violation acts as a ``witness'' attesting to the quantum (in fact, non-local) nature of the devices used to violate the inequality. This can then be used to derive conclusions regarding, e.g., the structure of the underlying quantum states~\cite{summers1987bell,popescu1992states,coladangelo2016all} or the randomness of the measurements' outcomes~\cite{colbeck2009quantum,vazirani2012certifiable}.
		
		We remark that by using the DI method we do not only treat imperfections in the quantum apparatuses which are known or can be characterised in advance. The DI approach allows one to derive conclusions without making assumptions regarding the types of imperfections. Even more drastically, the devices can be assumed to be ``malicious''\footnote{``Malicious devices'' should be understood here as devices that ``try to convince'' a verifier that they produce highly entangled states while this is not the case. This is relevant, e.g., in the context of benchmarking, where one may argue, intentionally or not, that a given source is better than another even though this is not true.}; as long as a violation of a Bell inequality is observed, the observer, or verifier, can be sure of the quantum nature of the systems without placing significant trust in the manufacture of the devices.

		Most previous works, both theoretical and experimental, that can be seen as DI entanglement certification (DIEC) procedures work only under the IID assumption and thus fail to be operational in the sense defined above.\footnote{This statement also holds for works in the semi-DI setting; see, e.g.,~\cite{cavalcanti2016quantum,mccutcheon2016experimental}.} This includes tests concerned with the demonstration of entanglement via an entanglement witness (i.e. answering a yes-no question)~\cite{guehne2009entanglement,barreiro2013device,bancal2014device}, as well as more quantitative analyses of entanglement measures such as  the negativity  and dimension witness~\cite{moroder2013device}. In all of these works, the focus is on relating properties of a single multipartite state to some asymptotic statistics. An application of these works in an experimental setting is then most straightforwardly obtained under the assumption that the experiment consists in a repetition of identical rounds in which the same state is produced independently each time, i.e., under the IID assumption. The same assumption is inherent to claims regarding the amount of remaining entanglement which was not consumed by the estimation procedure: such claims rely on the assumption that more states, identical to those which were used for the testing phase, can be created by the source.

		Another line of recent works deals with self-testing of high dimensional entangled states~\cite{mckague2016self,chao2016test,coladangelo2017parallel,coudron2016parallel,natarajan2017quantum,coladangelo2017robust}. The goal of such works is more ambitious than entanglement certification; they aim to quantify the distance of the state used to perform the relevant tests from some specific target state, e.g., $\ket{\Phi^L}$ (up to local isometries, which cannot be excluded in the DI setting). 
		Of course, once such a bound is derived it also implies lower bounds on various (continuous) entanglement measures. 
		
		A crucial disadvantage of these works is that they are not noise-tolerant in any realistic sense and therefore are not adequate when dealing with noisy sources of entanglement. For example, IID noisy sources, which produce many IID copies of a noisy entangled state, e.g., the two-qubit Werner state $\sigma = (1-\xi) \ket{\phi^+}\bra{\phi^+} + \xi\mathbb{I}/4$ for some constant (i.e., independent of the number of copies being created) noise value $\xi\in[0,1]$, pass the considered self-tests only with negligible probability. Thus, no conclusion regarding the amount of entanglement produced by such IID noisy sources can be derived.\footnote{Note that this issue is inherent to the distance measures used in all self-testing works (i.e., their objective) rather than some non-optimal properties of the specific tests considered in the mentioned works.} Furthermore, here as well, most self-testing results can only be used to describe the entanglement which was already consumed in the self-test. As far as we are aware, the only self-testing works where this is not the case are~\cite{chao2016test,natarajan2017quantum}.

	\subsection{Distillable entanglement}

		There are many different ways of quantifying entanglement, some of which were already mentioned above (for surveys see~\cite{horodecki2009quantum,plenio2014introduction}). One of the most basic and meaningful measures of entanglement is the \emph{distillable entanglement}. 
		Roughly speaking, given a multipartite state, its distillable entanglement describes the number of maximally entangled states that can be ``extracted'' out of it by employing only local operations and classical communication (LOCC). Since LOCC cannot be used to create entanglement between the parties, such a process, termed  entanglement distillation, indeed quantifies the entanglement of the initial state itself. 
		
		The task of entanglement distillation was first considered in~\cite{bennett1996purification,bennett1996concentrating,bennett1996mixed}. There, two parties, Alice and Bob, share $n$ copies of a bipartite mixed states $\sigma$. Their goal is to apply some LOCC to create $r<n$ copies of, say, $\ket{\Phi^+}$.
		This motivates the following, widely used, definition of distillable entanglement:
		\begin{defn}[Asymptotic IID distillable entanglement]
			The asymptotic IID distillable entanglement of a bipartite state $\sigma\in\mathcal{H}_A \otimes \mathcal{H}_B$ is given by
			\begin{equation}\label{eq:asym_iid_E_D}
				E_D^{\infty}(\sigma) = \sup \left\{ r \big| \lim_{n\rightarrow\infty} \left(  \sup_{\Gamma} F \left( \Gamma(\sigma^{\otimes n}), \Phi^{2^{r n}} \right) =1 \right) \right\} \;,
			\end{equation}
			where $\Gamma$ is an LOCC map (with respect to the bipartition of $\sigma$) and $F$ is the fidelity.
		\end{defn}

		Equation~\eqref{eq:asym_iid_E_D} describes a scenario in which one starts with $n$ independent copies of $\sigma$ and requires that, as $n$ goes to infinity, the error of the distillation protocol goes to zero.
		This explains why we call it here the \emph{asymptotic IID} distillable entanglement and not simply the distillable entanglement as it is usually called in the literature.
		However, as we claimed above, the sources that we consider do not necessarily produce IID states $\sigma^{\otimes n}$ (and they, definitely, do not emit $n\rightarrow\infty$ entangled states). Hence, $E_D^{\infty}(\sigma)$ does not truly quantify the entanglement produced in our scenario. 
		
		For our purpose, a more suitable entanglement measure is the \emph{one-shot distillable entanglement}~\cite{buscemi2010distilling}.  
		In the one-shot scenario Alice and Bob share a single copy of a bipartite state $\rho$ and their goal is to convert it to the maximally entangled state $\ket{\Phi^L}$, for the maximal  possible value $L$, using only LOCC. 
		We say that the distillation protocol is successful when the resulting state is $\varepsilon$-close to $\ket{\Phi^L}$ for some fixed~$\varepsilon$. 
		
		To state the definition of the one-shot distillable entanglement in a way comparable to its asymptotic IID counterpart $E_D^{\infty}$, we consider $\rho\in \mathcal{H}_A^{\otimes n} \otimes \mathcal{H}_B^{\otimes n}$ for a fixed $n$ and some Hilbert spaces $\mathcal{H}_A$ and $\mathcal{H}_B$ (while $\rho$ itself does not necessarily have the IID form $\sigma^{\otimes n}$). We then identify the one-shot distillation rate as $r=\log(L)/ n$.\footnote{For a given $\rho$ there are different ways of choosing $n$, $\mathcal{H}_A$ and $\mathcal{H}_B$ such that $\rho\in \mathcal{H}_A^{\otimes n} \otimes \mathcal{H}_B^{\otimes n}$. Thus, different choices can lead to different distillation \emph{rates} $r=\log(L)/ n$ (while $n \; E_D^{n,\varepsilon}(\rho)$ is independent of these choices). Later on there will be no ambiguity regarding the value of $n$ and hence the used definition will be the one most relevant for us.} 
		We can now use the following definition:
		\begin{defn}[One-shot distillable entanglement]
			Let $n\in \mathbb{N_+} $ and $\varepsilon\in [0,1]$. The one-shot distillable entanglement of a bipartite state $\rho\in \mathcal{H}_A^{\otimes n} \otimes \mathcal{H}_B^{\otimes n}$ is given by
			\begin{equation}\label{eq:one_shot_E_D}
				E_D^{n,\varepsilon}(\rho) = \sup \left\{ \log(L)/ n \big| \sup_{\Gamma} F\left(   \Gamma(\rho),  \Phi^{L} \right) =1-\varepsilon  \right\} \;,
			\end{equation}
			where $\Gamma$ is an LOCC map (with respect to the bipartition of $\rho$) and $F$ is the fidelity.
		\end{defn}

		$E_D^{n,\varepsilon}(\rho)$ describes the number of maximally entangled states which can be extracted, using LOCC, from a \emph{single} copy of an arbitrary bipartite state $\rho$ while allowing for some error $\varepsilon$. Hence, it captures the amount of entanglement available in $\rho$ when using it as a resource in quantum information processing tasks.

		\paragraph*{Structure of the paper.} 
		The following sections are arranged as follows. 
		In Section~\ref{sec:contribution} we present our contribution: our definition of an operational DIEC, the considered setting, and our results (the protocol and achieved rates). 
		One can find all the necessary preliminary information and notation in Section~\ref{sec:prelim}.
		Section~\ref{sec:main_proofs} is devoted to presenting the main steps of the proof and Section~\ref{sec:open_quest} includes several remaining open questions.
		All the technical details of the proofs can be found in the appendix.

\section{Our contribution: operational DIEC}\label{sec:contribution}
	
	After setting the stage in the previous section, we are now ready to state the objective of the current work -- DI certification of one-shot distillable entanglement -- and our results. 
	We start in Section~\ref{sec:goal} by introducing and motivating our definition of an operational DIEC protocol. We then explain in Section~\ref{sec:setting} the exact setting of the source and measurement devices considered in our work. In Section~\ref{sec:results} we present our results. 
	
	\subsection{The goal}\label{sec:goal}
			
		Let us start by defining explicitly  what we mean by a DIEC protocol and, by this, set the goal of this work.
		Given an uncharacterised source of entanglement producing $n$ bipartite systems globally described by the  state $\phi\in\mathcal{H}_{\tilde{A}}^{\otimes n} \otimes \mathcal{H}_{\tilde{B}}^{\otimes n}$ for some (unknown) Hilbert spaces $\mathcal{H}_A$ and $\mathcal{H}_B$ and (at least) two measurement devices not entangled to one another but otherwise uncharacterised, our goal is to find a DIEC protocol, employing only LOCC, that certifies that $\phi$ is highly entangled in a meaningful operational way. The certification protocol is going to act on $\phi$, using the measurement devices, and we would  like to claim (roughly) that, if the protocol does not abort, then the \emph{final state} has high amount of one-shot distillable entanglement. 
		
		More precisely, we define a DIEC as follows.
		\begin{defn}[DIEC protocol]\label{def:DIEC}
			For any $n\in \mathbb{N_+} $ let $\varepsilon_{\mathrm{dist}},\varepsilon_{\mathrm{snd}},\varepsilon_{\mathrm{cmp}}\in [0,1]$ be the desired error constants, and $r\in [0,1]$ the desired threshold distillation rate. 
			Furthermore, let $\mathcal{S}^{\mathrm{honest}}$ be a set of states $\phi^{\mathrm{honest}}$ produced by a desired ``honest source'' and $\mathcal{D}^{\mathrm{honest}}$ the set of the desired ``honest measurement devices''.\footnote{Mathematically, this set can be defined as, e.g., a set of projectors applied by the measurement devices for the different possible inputs.}
			Let $\mathrm{P}$ be a protocol \emph{employing only} LOCC that given a state $\phi\in\mathcal{H}_{\tilde{A}}^{\otimes n} \otimes \mathcal{H}_{\tilde{B}}^{\otimes n}$ creates a state $\rho\in\mathcal{H}_{A}^{\otimes n} \otimes \mathcal{H}_{B}^{\otimes n}$. We denote the final state conditioned on not aborting the protocol by $\rho_{|\Omega}$.

			The protocol $\mathrm{P}$ is said to be a DIEC protocol if the following two properties hold:
			\begin{enumerate}
				\item Noise-tolerance (completeness): The probability that $\mathrm{P}$ aborts when applied on any $\phi^{\mathrm{honest}}\in\mathcal{S}^{\mathrm{honest}}$ using honest measurement devices from $\mathcal{D}^{\mathrm{honest}}$  is at most~$\varepsilon_{\mathrm{cmp}}$.
				\item Entanglement certification (soundness): For any source and measurement devices either $\mathrm{P}$ aborts with probability greater than $1-\varepsilon_{\mathrm{snd}}$ when applied on $\phi$ or $E_D^{n,\varepsilon_{\mathrm{dist}}}(\rho_{|\Omega}) \geq r$.
			\end{enumerate}
		\end{defn}	
		\hspace{1pt}
		
		There are several important remarks to make regarding the above definition. 
		\begin{enumerate}
			\item One possible example for an honest source is a source producing $n$ independent copies of the Werner state $\sigma = (1-\xi) \ket{\phi^+}\bra{\phi^+} + \xi\mathbb{I}/4$ for some maximal value of $\xi>0$. 
			We can then define  
			\[
				\mathcal{S}^{\mathrm{honest}} = \left\{ \left( (1-\tilde{\xi}) \ket{\phi^+}\bra{\phi^+} + \tilde{\xi}\mathbb{I}/4\right)^{\otimes n} \; | \; \tilde{\xi} \leq \xi \; \right\} \;.
			\]
			$\mathcal{D}^{\mathrm{honest}}$ can be defined to include, e.g., the measurement devices that apply the optimal measurements performed in the CHSH game. 
			The noise-tolerance property, also termed completeness, states that the protocol should not abort, with high probability, for any $\phi^{\mathrm{honest}}\in\mathcal{S}^{\mathrm{honest}}$ even though some noise $\xi>0$ is present. 
			This implies (in combination with the soundness property) that $\mathrm{P}$ is able to certify the entanglement produced by the honest source. 
			Of course, the honest sets $\mathcal{S}^{\mathrm{honest}}$ and $\mathcal{D}^{\mathrm{honest}}$ can be chosen in any way one wishes and depending on the experimental setting one has in mind. For instance, it could include states with a different noise model than the Werner state.
			In most cases the manufacture of the entanglement source (the experimentalist) has some ``guess'' for a realistic description of the source and measurement devices. In most applications, these define the sets that should be chosen as the honest sets.
			
			\item The entanglement certification, or soundness, property means that for $any$ state produced by the unknown source and $any$ measurement devices either the protocol identifies that the apparatus is not sufficiently good, and therefore the protocol aborts, or the post-protocol state $\rho_{|\Omega}$ is highly entangled, in the sense that there exists a way to distill (close to) $nr$ maximally entangled states out of it.
			The protocol therefore certifies that $\rho_{|\Omega}$ is indeed useful for subsequent applications.\footnote{Note that below we only claim that a distillation protocol \emph{exists} but we do not present it explicitly.} 
			
			\item To be able to quantify the entanglement produced by the \emph{source itself}, rather than the entanglement that could be present inside the measurement devices, we need to assume that the measurement devices were not entangled before the start of the protocol. This situation is somehow similar to the requirement for measurement independence in a Bell test: if the measurement settings used in a Bell test are chosen by a device correlated with the source, then quantum statistics can be reproduced by local models~\cite{hall2011relaxed,barrett2011how}. In practice, measurement independence is guaranteed by assuming that the devices used to produce the measurement settings behave independently from the rest of the setup. Similarly, we could assume that the measurement device of each party is independent from all other devices involved in the protocol. However, this hypothesis is stronger than necessary. We thus make the lighter assumption that the measurement devices share no entanglement with each other at the beginning of the protocol.

			\item Only LOCC protocols should be considered as DIEC protocols since we would like to certify the distillable entanglement produced by the source and not the entanglement that may be produced by a protocol employing operations that cannot be explained via LOCC.\footnote{Another option may be to consider protocols that employ only separability preserving operations; see~\cite{rains1997entanglement,cirac2001entangling,brandao2011one}.}
			
			\item In a way, the definition of a DIEC can be seen as an extension of the so called SWAP technique~\cite{bancal2015physical}, used in self-testing works, to the non-IID setting. Roughly speaking, by constructing the SWAP operators one can claim that a state close to, e.g., the maximally entangled state, up to local isometries can be ``extracted'' out of the uncharacterised devices. The construction of a DIEC protocol implies that one can extract a state close to \emph{many} maximally entangled states using LOCC.
			
			\item The above definition allows us to use DIEC protocols as a way to \emph{compare} different entanglement sources which are otherwise incomparable. For example, one can consider two sources developed by different experimental groups. One source produces, say, many identical copies of the Werner state while the other creates many identical copies of perfect partially entangled states. 
			Each group is free to choose a DIEC protocol, as in Definition~\ref{def:DIEC}, that will result in the highest lower bound on the distillable entanglement produced by its source. 
			If one group wishes to claim that its source is ``better'', then the certified distillation rate of their source (which can be verified by any user) must be higher than that of the competing source. 
			The protocols allow us to \emph{benchmark} one source against the other in a meaningful way by allowing any user to verify a claimed lower bound on the produced entanglement available for further applications. 
		\end{enumerate}

	\subsection{The setting: source and measurement devices}\label{sec:setting}
	
		Here we describe the theoretical setting considered in our work to which our results apply, i.e., under which we will prove the soundness of our DIEC protocol.
		Different variants of this setting can be chosen depending on one's interest; we choose the presented one as we believe it is both realistic and simple to discuss. 
		This setting is compatible with the standard assumptions used in the DI setting (e.g., when testing the CHSH Bell inequality in DI quantum key distribution or randomness generation protocols).
		In particular, we employ the following standard assumptions: the measurement devices are separated in space and are the verifier can restrict their communication in different stages of the protocol and the verifier holds a trusted random number generator and is able to make basic classical calculations required to run the protocol. 
		In addition, we assume that quantum physics is correct.


		\subsubsection{The theoretical setting: source, measurement devices, and quantum registers}\label{sec:setting_theo}	
		
			In order to be able to talk about the entanglement available \emph{after} running the DIEC protocol we must be able to have a well defined state at hand, whose entanglement we are quantifying. For this, we consider a theoretical setting which fulfils the following  two conditions. 

			First, the production of any entangled states (if such are being created) can be attributed only to the source. In other words, we assume that the measurement devices \emph{neither produce nor hold} additional entanglement. This implies a distinction between the source and the measurement devices. 

			Second, the state produced by the source, or the ``post-protocol state'', can be kept in some registers, i.e., quantum memory. This is necessary at the theoretical level since we wish to discuss the remaining entanglement in a meaningful way. The registers are \emph{trusted}, in the sense that they cannot be manipulated by the devices after they are accessed during the protocol. Indeed, if the devices are allowed to, e.g., measure the registers in which the final state is being kept then, clearly, one cannot say anything regarding the entanglement left in the system.\footnote{This also fits the distinction we made above between the source and the measurement devices. If the entanglement is produced by the measurement devices and is kept ``inside of them'' then they can also destroy whatever is left in the end of the protocol.}

			Note that by considering the one-shot distillable entanglement we are already hinting that one should be able to apply an entanglement distillation protocol on the certified state (at least in theory). To apply such a protocol, the state must be available somewhere so it can be manipulated. Defining some quantum registers in which the state is being kept is therefore necessary in our context. 
			The quantum registers are merely a theoretical tool, they are not needed in an experimental implementation of our DIEC protocol.

			The reminder of the current section is devoted to explaining precisely the setting, fulfilling the above conditions that we consider in the current work. The setting is illustrated in Figure~\ref{fig:source_mod}.

			\begin{figure}
				\centering
				\includegraphics[width=\textwidth]{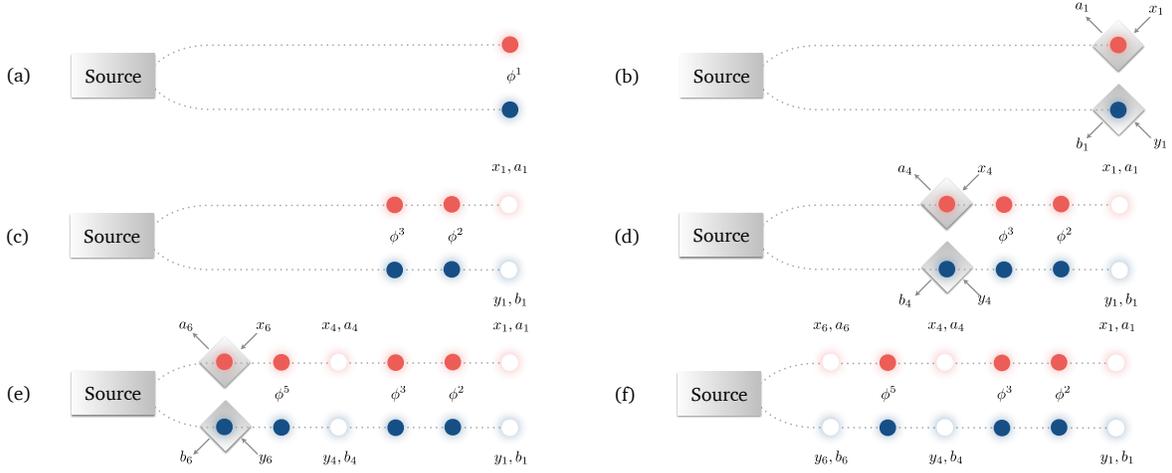}
				\caption{Illustration of the considered setting. (a)~The source emits the first state in the sequence $\phi^1$. (b)~The verifier chooses to measure the first state. The two measurement devices, denoted by the grey squares, measure $\phi^1$ according to the inputs $x_1,y_1$ given by the verifier. They output $a_1,b_1$. (c)~The next states $\phi^2,\phi^3$ are produced by the source one after the other. They can depend on the values $x_1,y_1,a_1,b_1$. (d)~The source emits $\phi^4$ and the verifier chooses to measure it using his measurement devices. The actions of the measurement devices can depend on the classical information of the previous rounds: the values $x_1,y_1,a_1,b_1$ as well as the knowledge that $\phi^2,\phi^3$ were not measured. (e)~The next states are produced by the source. The verifier chooses to measure only $\phi^6$. (f)~The final state which is kept in the memory. 
				}\label{fig:source_mod}
			\end{figure}

			The entity wanting to certify the entanglement produced by the source is called the \emph{verifier}. 
			To run the DIEC protocol the verifier holds two separated (space-like or otherwise shielded) measurement devices. As we consider two measurement devices we can also treat the verifier as two parties, Alice and Bob (both the ``verifier'' and ``Alice and Bob'' are used below, depending on the context).
		
			We consider an uncharacterised source emitting a sequence of entangled quantum systems . (For example, one can imagine a source that emits pairs of entangled photons one after the other).
			We denote by $n\in \mathbb{N}_+$ the number of systems  produced by the source, i.e., the number of times the source is being used during the DIEC protocol. 
			For every $i\in[n]$, the source produces some unknown state~$\phi^i$ that is being kept in the memory.
			The verifier then chooses whether to measure the state using his measurement devices or not. 
			If he measures the state then the classical inputs given to the devices and the outputs produced by them are kept in a classical  memory. If $\phi^i$ is not measured then it is kept as is in a quantum memory.

			The next state produced by the source,  $\phi^{i+1}$, can  depend on all of the previous classical information, i.e., which of the previous states were measured and what were the inputs and outputs in those rounds.
			Note that while $\phi^{i+1}$ can depend on whether a test was made in a previous round $j<i+1$, the source and the measurement devices have no further access to the quantum states $\phi^j$ which \emph{were not measured}. That is, we assume that each $\phi^j$ which was not measured is not affected by the devices after the $j$'th round.
			Again, this is necessary as otherwise the remaining entanglement can be destroyed by the devices. We emphasise that this \emph{does not imply} that the systems of the different rounds cannot be entangled with one another.

			We use the above model in our protocol and its analysis. However, we remark that the protocol and analysis can be adapted to capture other models with a sequential structure.


	\subsection{Results}\label{sec:results}
	
		\subsubsection{Our DIEC protocol and the achieved rates}\label{sec:diec_result}

			Our DIEC protocol is presented as Protocol~\ref{pro:ent_test}. It is based on the CHSH inequality (see Section~\ref{sec:pre_chsh} for the necessary basic information). Similar protocols can also be considered for other Bell inequalities. 
	
			\begin{algorithm}[t]
			\caption{DIEC protocol (based on the CHSH inequality)}
			\label{pro:ent_test}
			\begin{algorithmic}[1]
				\STATEx \textbf{Arguments:} 
					\STATEx\hspace{\algorithmicindent} $D$ -- untrusted measurement device of two components with inputs and outputs set $\{0,1\}$
					\STATEx\hspace{\algorithmicindent} $n \in \mathbb{N}_+$ -- number of rounds
					\STATEx\hspace{\algorithmicindent} $\gamma$ -- the probability of a test
					\STATEx\hspace{\algorithmicindent} $\omega_{\mathrm{exp}}$ -- expected winning probability in the CHSH game  
					\STATEx\hspace{\algorithmicindent} $\delta_{\mathrm{est}} \in (0,1)$ -- width of the statistical confidence interval for the estimation test
					
				\STATEx
				
				\STATE For every round $i\in[n]$ do Steps~\ref{prostep:first}-\ref{prostep:last}:
					\STATE\hspace{\algorithmicindent}Let $\phi^i$ denote the bipartite state produced by the source in this round. \label{prostep:first}
					\STATE\hspace{\algorithmicindent}Set $A_i,B_i,X_i,Y_i,W_i=\perp$.
					\STATE\hspace{\algorithmicindent}Choose $T_i=1$ with probability $\gamma$ and $T_i=0$ otherwise.
					\STATE\hspace{\algorithmicindent}If $T_i=1$: 
						\STATE\hspace{\algorithmicindent}\hspace{\algorithmicindent}Choose the inputs $X_i,Y_i\in\{0,1\}$ uniformly at random.
						\STATE\hspace{\algorithmicindent}\hspace{\algorithmicindent}Measure $\phi^i$ using $D$ with the inputs $X_i,Y_i$ and record the outputs as $A_i,B_i\in\{0,1\}$. \label{prostep:measurement} 
						\STATE\hspace{\algorithmicindent}\hspace{\algorithmicindent}Set $W_i = 1$ if the CHSH game is won and $W_i=0$ otherwise.
					\STATE\hspace{\algorithmicindent}If $T_i=0$:
						\STATE\hspace{\algorithmicindent}\hspace{\algorithmicindent}Keep $\phi^i$ in the registers $\hat{A}_i\hat{B}_i$. \label{prostep:last}
					\STATE Abort if $\sum_i \chi(T_i=1)\,W_i < (\omega_{\mathrm{exp}}\gamma - \delta_{\mathrm{est}}) n \;$. \label{prostep:abort_ET_1}
			\end{algorithmic}
			\end{algorithm}
	
			As mentioned in the previous section, we consider a source which produces a sequence of bipartite states, denoted by $\phi^i$ for $i\in[n]$. 
			The verifier chooses whether to measure $\phi^i$ or keep it as is in the memory. The register $T_i$ describes whether a test was performed or not. 
			If a test is performed, the registers $X_iY_iA_iB_i$ hold the classical inputs and outputs. The register $W_i$ is set to $1$ when the CHSH game is won in the $i$'th round and $0$ otherwise. 
			When a test is not being performed, the state $\phi^i$ is kept in the quantum registers $\hat{A}_i\hat{B}_i$.
			We allow the source to ``know'' all the classical information of the previous rounds, i.e., $(TABXYW)_{1,\dots,i-1}$ and hence $\phi^i$ can also depend on this information.
			
			We denote the state in the end of Protocol~\ref{pro:ent_test}, before Step~\ref{prostep:abort_ET_1}, by 
			\begin{equation*}
				\rho = \rho_{\hat{A}\hat{B}ABXYTW} \;.
			\end{equation*}
			
			Denote by $\chi$ the indicator function, i.e. $\chi=1$ if $T_i=1$, and $\chi=0$ otherwise.
			Let $\Omega$ be the event $\sum_i \chi(T_i=1)\,W_i \geq (\omega_{\mathrm{exp}}\gamma - \delta_{\mathrm{est}}) n \;$, i.e., the event of not aborting the protocol. 
			We denote the state after the end of the protocol, conditioned on not aborting, by
			\begin{equation*}
				\rho_{|\Omega} = \rho_{\hat{A}\hat{B}ABXYTW|\Omega} \;.
			\end{equation*}
	
			We define the bipartition of $\rho$ and $\rho_{|\Omega}$ as $Q^A=\hat{A}$ vs.\@ $Q^B= \hat{B}ABXYTW$. 
			It is easy to see that for this bipartition Protocol~\ref{pro:ent_test} employs only LOCC for \emph{any} measurement devices in our setting (i.e., initially non-entangled devices that do not communicate during the time of the measurement).
	
			Our main result states that Protocol~\ref{pro:ent_test} is indeed an operational DIEC protocol. That is, it fulfils the requirements of Definition~\ref{def:DIEC}.

			As the honest source and devices we choose to consider a source that produces identical and independent copies of a state $\phi^i=\sigma$ and  measurement devices that apply the same measurements in each round when they are used. The state $\sigma$ and the measurements are such that the winning probability achieved in the CHSH game is at least $\omega_{\mathrm{exp}}$. 
			For example, one can choose $\mathcal{D}^{\mathrm{honest}}$ to include the measurement devices that apply the optimal measurements performed in the CHSH game and the
			 set $\mathcal{S}^{\mathrm{honest}}$ to include all states $\phi^{\mathrm{honest}}=\sigma^{\otimes n}$ for~$\sigma$ any noisy maximally entangled state that will result in winning probability $\geq \omega_{\mathrm{exp}}$.

			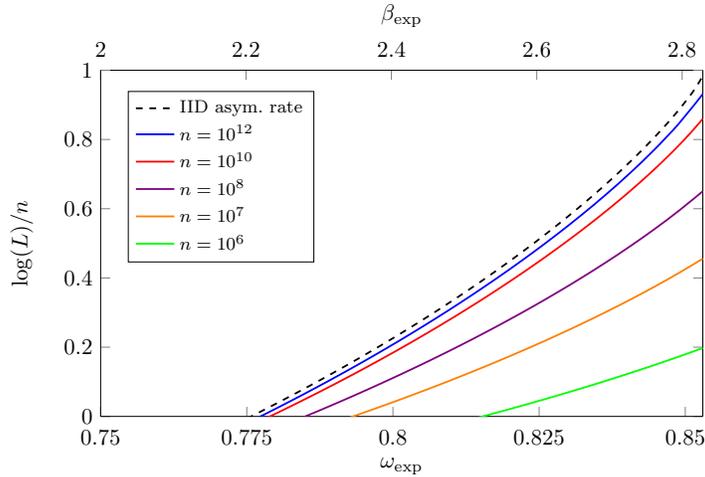
\begin{figure}
			\centering
			\begin{tikzpicture}[scale=0.85]
				\begin{axis}[
					height=7cm,
					width=11cm,
					xlabel=$\beta_{\mathrm{exp}}$,
					xmin=2,
					xmax=2.82842,
				    	xtick={2,2.2,2.4,2.6,2.8},
					x tick label style={yshift={2.5pt}},
					axis x line*=top,
					xlabel near ticks,
					ymax=1,
					ymin=0,
					hide y axis
				]
				\end{axis}
				\begin{axis}[
					height=7cm,
					width=11cm,
					xlabel=$\omega_{\mathrm{exp}}$,
					ylabel=$\log(L)/n$,
					xmin=0.75,
					xmax=0.853,
					ymax=1,
					ymin=0,
				    	xtick={0.75,0.775,0.80,0.825,0.85},
					x tick label style={yshift={-2.5pt},  
						/pgf/number format/.cd, precision=3, /tikz/.cd},
			         	ytick={0,0.2,0.4,0.6,0.8,1},
					axis x line*=bottom,
					legend style={at={(0.2,0.94)},anchor=north,legend cell align=left,font=\footnotesize} 
				]
				
			
				\addplot[black,thick,smooth,dashed] coordinates {
				(0.775013, -0.00628557) (0.776684, 0.00797869) (0.778355, 0.0224093) (0.780026, 0.0370094) (0.781697, 0.0517826) (0.783369, 0.0667322) (0.78504, 0.0818622) (0.786711, 0.0971763) (0.788382, 0.112679) (0.790053, 0.128374) (0.791724, 0.144266) (0.793395, 0.16036) (0.795066, 0.176661) (0.796737, 0.193175) (0.798408, 0.209907) (0.800079, 0.226862) (0.80175, 0.244049) (0.803421, 0.261473) (0.805092, 0.279141) (0.806763, 0.297062) (0.808435, 0.315244) (0.810106, 0.333696) (0.811777, 0.352427) (0.813448, 0.371448) (0.815119, 0.39077) (0.81679, 0.410406) (0.818461, 0.430369) (0.820132, 0.450674) (0.821803, 0.471337) (0.823474, 0.492376) (0.825145, 0.51381) (0.826816, 0.535663) (0.828487, 0.557959) (0.830158, 0.580727) (0.83183, 0.603998) (0.833501, 0.627812) (0.835172, 0.652211) (0.836843, 0.677247) (0.838514, 0.702982) (0.840185, 0.729494) (0.841856, 0.756876) (0.843527, 0.785253) (0.845198, 0.814787) (0.846869, 0.845711) (0.84854, 0.878374) (0.850211, 0.91337) (0.851882, 0.95195) (0.853553, 1)	
				};
				\addlegendentry{IID asym.\@ rate}
				
				\addplot[blue,thick,smooth] coordinates {
				(0.77, -0.060422) (0.77415, -0.0266105) (0.7783, 0.00817149) (0.78245, 0.0439699) (0.7866, 0.0808357) (0.79075, 0.118826) (0.7949, 0.158007) (0.79905, 0.19845) (0.8032, 0.240243) (0.80735, 0.283482) (0.8115, 0.328285) (0.81565, 0.37479) (0.8198, 0.423165) (0.82395, 0.473617) (0.8281, 0.526409) (0.83225, 0.581888) (0.8364, 0.640525) (0.84055, 0.703009) (0.8447, 0.770431) (0.84885, 0.844809) (0.853, 0.931351)
				};
				\addlegendentry{$n=10^{12}$}
			
				\addplot[red,thick,smooth] coordinates {
				(0.7783, -0.00527627) (0.780375, 0.0115888) (0.78245, 0.028699) (0.784525, 0.0460603) (0.7866, 0.0636792) (0.788675, 0.0815624) (0.79075, 0.0997171) (0.792825, 0.118151) (0.7949, 0.136872) (0.796975, 0.155889) (0.79905, 0.175211) (0.801125, 0.194848) (0.8032, 0.214812) (0.805275, 0.235113) (0.80735, 0.255764) (0.809425, 0.276778) (0.8115, 0.298171) (0.813575, 0.319958) (0.81565, 0.342158) (0.817725, 0.364789) (0.8198, 0.387874) (0.821875, 0.411436) (0.82395, 0.435502) (0.826025, 0.460102) (0.8281, 0.485272) (0.830175, 0.51105) (0.83225, 0.537482) (0.834325, 0.564622) (0.8364, 0.592533) (0.838475, 0.621291) (0.84055, 0.650991) (0.842625, 0.68175) (0.8447, 0.713722) (0.846775, 0.747115) (0.84885, 0.782224) (0.850925, 0.819504) (0.853, 0.859758)
				};
				\addlegendentry{$n=10^{10}$}
				
				\addplot[violet,thick,smooth] coordinates {
				(0.78245, -0.0177506) (0.7866, 0.01127870) (0.79075, 0.041160627) (0.7949, 0.07193876) (0.79905, 0.10366215) (0.8032, 0.13638638) (0.80735, 0.170174940) (0.8115, 0.205101022) (0.81565, 0.24124991) (0.8198, 0.278722214) (0.82395, 0.31763837) (0.8281, 0.35814516) (0.832249999, 0.400425285) (0.8364, 0.44471246) (0.84055, 0.49131627) (0.8447, 0.54066632) (0.84885, 0.5933995) (0.853, 0.65056124)
				};
				\addlegendentry{$n=10^8$}
			
				\addplot[orange,thick,smooth] coordinates {
				(0.79075, -0.0140405) (0.7949, 0.0102783) (0.79905, 0.0353165) (0.8032, 0.0611116) (0.80735, 0.0877063) (0.8115, 0.115149) (0.81565, 0.143494) (0.8198, 0.172806) (0.82395, 0.203161) (0.8281, 0.234646) (0.83225, 0.267368) (0.8364, 0.301459) (0.84055, 0.337083) (0.8447, 0.374454) (0.84885, 0.413858) (0.853, 0.4557)
				};
				\addlegendentry{$n=10^7$}
				
				\addplot[green,thick,smooth] coordinates {
				(0.8115, -0.0170334) (0.81565, 0.00141505) (0.8198, 0.0204221) (0.82395, 0.0400194) (0.8281, 0.0602427) (0.83225, 0.081133) (0.8364, 0.102738) (0.84055, 0.125113) (0.8447, 0.148325) (0.84885, 0.172452) (0.853, 0.197594)
				};
				\addlegendentry{$n=10^6$}

				\end{axis}  
			\end{tikzpicture}
			
			\caption{The one-shot distillable entanglement rate $r=\log L/n$ as a function of the expected winning probability in the CHSH game $\omega_{\mathrm{exp}}$ (or $\beta_{\mathrm{exp}}$ for different values of $n$ (the analytical form is given in Equation~\eqref{eq:log_l_intro}). The error parameters were chosen to be $\varepsilon_{\mathrm{dist}}=\varepsilon_{\mathrm{snd}}=10^{-5}$ and the completeness parameter was set to $\varepsilon_{\mathrm{cmp}}  =10^{-2}$. All other values were chosen so that the distillable entanglement rate is maximised. The dashed curve describes the distillable entanglement rate which can be certified in the IID asymptotic case.}\label{fig:dist_ent_rate}
			\end{figure}

			The following theorem states our main result.
			\begin{thm}[Main theorem]\label{thm:main_thm}
				For any $n\in \mathbb{N_+} $, $\varepsilon_{\mathrm{dist}},\varepsilon_{\mathrm{snd}}\in [0,1]$, and $\varepsilon_{\mathrm{smo}}\in[0,\sqrt{\varepsilon_{\mathrm{dist}}})$, Protocol~\ref{pro:ent_test} is a DIEC protocol with:
				\begin{enumerate}
					\item Noise-tolerance (completeness): The probability that $\mathrm{P}$ aborts when applied on any $\phi^{\mathrm{honest}}\in\mathcal{S}^{\mathrm{honest}}$ using honest measurement devices from $\mathcal{D}^{\mathrm{honest}}$  is at most $\varepsilon_{\mathrm{cmp}} \leq \exp(- 2 n\delta_{\mathrm{est}}^2 )$.
					\item Entanglement certification (soundness): For \emph{any} source and measurement devices in the considered setting, either Protocol~\ref{pro:ent_test} aborts with probability greater than $1-\varepsilon_{\mathrm{snd}}$ when applied on $\phi$ or 
					$E_D^{n,\varepsilon_{\mathrm{dist}}}(\rho_{|\Omega}) \geq r$ for $r=\log(L)/ n$ and
					\begin{equation}\label{eq:log_l_intro}
						\log L= - n\cdot \eta_{\mathrm{opt}}(\varepsilon_{\mathrm{smo}},\varepsilon_{\mathrm{snd}}) -4\log\left(  \frac{1}{\sqrt{\varepsilon_{\mathrm{dist}}}-\varepsilon_{\mathrm{smo}}}\right)  \in \Omega(n) 
					\end{equation}
					where  $\eta_{\mathrm{opt}}$ is defined in Equation~\eqref{eq:eta_opt}. 
				\end{enumerate}
			\end{thm}

			The exact bound on the certified distillation rate $r$ is not very informative, so we postpone discussing its explicit form to later. 
			Instead, we plot the distillation rate in Figure~\ref{fig:dist_ent_rate} for some choices of parameters. 
			Clearly, the optimal distillation rate is upper bounded by $1$ and hence $\log L \in \Omega(n)$, as achieved in our work, is optimal in terms of the asymptotic dependency on $n$. 
			
			To discuss the tightness of the derived rates more concretely we need to first say few words about our proof technique. 
			To prove Theorem~\ref{thm:main_thm} we use two independent results.
			The first is a lower bound on the one-shot distillable entanglement in terms of the negative 
			of the \emph{conditional smooth max-entropy} $H_{\max}^{\varepsilon_{\mathrm{smo}}}(Q^A|Q^B)_{\rho_{|\Omega}}$~\cite{wilde2016converse}. This allows us to reduce the task of proving Theorem~\ref{thm:main_thm} to that of (upper) bounding $H_{\max}^{\varepsilon_{\mathrm{smo}}}(Q^A|Q^B)_{\rho_{|\Omega}}$. 
			The main tool used to derive a bound on the smooth max-entropy is the entropy accumulation theorem~\cite{dupuis2016entropy}.
			
			Our bound on  $H_{\max}^{\varepsilon_{\mathrm{smo}}}(Q^A|Q^B)_{\rho_{|\Omega}}$ is tight to first order in $n$ and hence cannot be significantly improved. 
			In particular, this means that the regime of Bell violations $\omega_{\mathrm{exp}}\lessapprox 0.775$ in which our distillation rate is zero, although the verifier observes a violation, cannot be improved by better bounding $H_{\max}^{\varepsilon_{\mathrm{smo}}}(Q^A|Q^B)_{\rho_{|\Omega}}$.
			
			Qualitatively, it is not surprising that such a regime exists. Indeed, it was already shown that, in general, Bell non-locality is fundamentally different from distillable entanglement (this stands in contrast to the so called Peres conjecture). That is, there are \emph{bound entangled states} (i.e., entangled state which cannot be distilled)  that can be used to violate some Bell inequalities~\cite{vertesi2014disproving}.
			For the CHSH inequality, however, this is not the case~\cite{masanes2006asymptotic}: bound entangled states cannot be used to violate the CHSH inequality. Hence, \emph{asymptotically}, one should be able to certify distillable entanglement for \emph{any} violation $\omega_{\mathrm{exp}}>0.75$. 
			
			This implies that, although our bound on  $H_{\max}^{\varepsilon_{\mathrm{smo}}}(Q^A|Q^B)_{\rho_{|\Omega}}$ is tight, the relation between the smooth max-entropy and the one-shot distillable entanglement is not; it only gives a \emph{lower bound} on the amount of distillable entanglement but it should be possible to distill more (see Section~\ref{sec:open_quest} for more details).

		\subsubsection{von Neumann entropy as a function of the CHSH violation}\label{sec:von_neumann_intro}
		
			As a step in our proof of Theorem~\ref{thm:main_thm} we derive an upper bound on the conditional von Neumann entropy $H(\hat{A}_i|\hat{B}_i)_{\sigma}$ for any Bell-diagonal quantum state $\sigma_{\hat{A}_i\hat{B}_i}$ that can be used to win the CHSH game with winning probability $\omega\in\left[\frac{3}{4},\frac{2+\sqrt{2}}{4}\right]$. The conditional von Neumann entropy is negative only when evaluated over entangled states. Thus, an upper bound on it can be seen, by itself, as a quantitive certificate of entanglement (though not in our operational sense). Such a bound might be of independent interest in other contexts. 
			
			Inspired by~\cite{pironio2009device}, in which a ``dual quantity''\footnote{In~\cite{pironio2009device} a bound on $H(A|E)$ was derived, where $A$ is the measurement outcome on Alice's side (when measuring~$\sigma_{\hat{A}_i}$ in our notation) and $E$ describes a system used to purify $\sigma_{\hat{A}_i\hat{B}_i}$.} was bounded, we derive a bound for Bell diagonal states. We prove the following:

			\begin{lemma}\label{lem:Bell_diag_entropy_intro}
				For any Bell diagonal state $\sigma_{\hat{A}_i\hat{B}_i}$ that can be used to violate the CHSH inequality with violation $\omega\in\left[\frac{3}{4},\frac{2+\sqrt{2}}{4}\right]$,
				\begin{equation}\label{eq:single_round_intro}
					H(\hat{A}_i|\hat{B}_i)_{\sigma} \leq 2h\left(\frac{1}{2}-\frac{2\omega -1}{\sqrt{2}}\right) - 1 \;,
				\end{equation}
				where $h$ is the binary entropy function.
			\end{lemma}

			The bound given in Equation~\eqref{eq:single_round_intro} is plotted in Figure~\ref{fig:single_round_entropy}. The bound is tight, i.e., there exist states that saturate it. 
			As mentioned above, the ``interesting'' regime of the bound is that in which the conditional entropy is negative, i.e. $\omega\gtrapprox 0.775$ ($\beta\gtrapprox 2.2$).
			A similar bound, but for the conditional max-entropy $H_{\max}(\hat{A}_i|\hat{B}_i)_{\sigma}$ was derived in~\cite[Supplementary Note~2]{pfister2016universal}. As $H_{\max}(\hat{A}_i|\hat{B}_i)_{\sigma} \geq H(\hat{A}_i|\hat{B}_i)_{\sigma} $, their bound can also be used to upper bound the von Neumann entropy. However, using our bound directly leads to better quantitive results. In particular, the bound derived in~\cite{pfister2016universal} only leads to a negative upper bound for $\beta\gtrapprox 2.5$.
			
			The observation that there exists a regime in which the conditional entropy is positive although the CHSH inequality is violated is not new. Indeed, it was already known that  that some states, e.g., the Werner state, can be used to violate the CHSH inequality while presenting positive conditional entropy~\cite{friis2017geometry}. 
			
			The fact that the bound~\eqref{eq:single_round_intro} on the conditional entropy is negative as soon as the winning probability is larger than $\omega\gtrapprox 0.775$ allows our scheme to certify entanglement for honest implementations based only on their CHSH winning probability. In particular, the honest implementation might be different than the Werner state we considered above as an example. At the same time, since the minimum winning probability required to obtain a useful bound on the conditional entropy is larger than 0.75, a significant Bell violation is required. This has implications on the critical detection efficiency of the scheme: whereas a violation of the CHSH Bell inequality is possible as soon as the detection efficiency is larger than $66.7\%$~\cite{eberhard93background}, the minimal detection efficiency required to reach a CHSH value of $2.2$ is~$85.3\%$.
			
			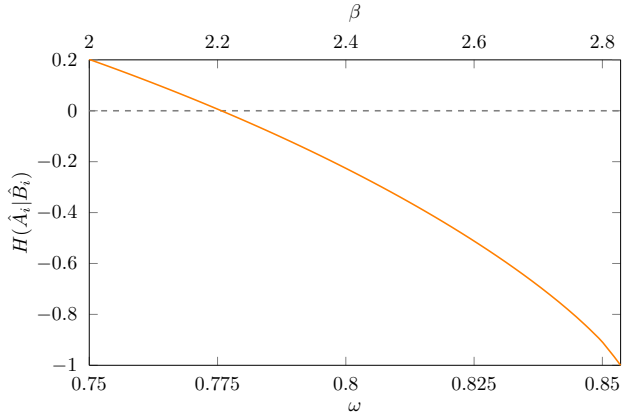
\begin{figure}
			\centering
			\begin{tikzpicture}[scale=0.75]
				\begin{axis}[
					height=7cm,
					width=11cm,
					xlabel=$\beta$,
					xmin=2,
					xmax=2.82842,
				    	xtick={2,2.2,2.4,2.6,2.8},
					x tick label style={yshift={2.5pt}},
					axis x line*=top,
					xlabel near ticks,
					ymax=0.2,
					ymin=-1,
					hide y axis
				]
				\end{axis}
				\begin{axis}[
					height=7cm,
					width=11cm,
					xlabel=$\omega$,
					ylabel=$H(\hat{A}_i|\hat{B}_i)$,
					xmin=0.75,
					xmax=0.853553,
					ymax=0.2,
					ymin=-1,
				    	xtick={0.75,0.775,0.80,0.825,0.85},
					x tick label style={yshift={-2.5pt},  
						/pgf/number format/.cd, precision=3, /tikz/.cd},
			         	ytick={0.2,0,-0.2,-0.4,-0.6,-0.8,-1},
					y tick label style={xshift={-2.5pt}},
					axis x line*=bottom
				]
				
				
				\addplot[gray,thick,dashed] coordinates {
				(0.75,0) (0.853553, 0)
				};
							
				\addplot[orange,thick,smooth] coordinates {
				(0.75, 0.201752) (0.753452, 0.176646) (0.756904, 0.150974) (0.760355, 0.124719) (0.763807, 0.0978629) (0.767259, 0.0703864) (0.770711, 0.042268) (0.774162, 0.0134847) (0.777614, -0.0159887) (0.781066, -0.0461796) (0.784518, -0.0771181) (0.78797, -0.108837) (0.791421, -0.141374) (0.794873, -0.174769) (0.798325, -0.209068) (0.801777, -0.244322) (0.805228, -0.280591) (0.80868, -0.31794) (0.812132, -0.356447) (0.815584, -0.396202) (0.819036, -0.437312) (0.822487, -0.479905) (0.825939, -0.524139) (0.829391, -0.570209) (0.832843, -0.61837) (0.836294, -0.668959) (0.839746, -0.722456) (0.843198, -0.779581) (0.84665, -0.841563) (0.850102, -0.910984) (0.853553, -1)
				};

				\end{axis}  
			\end{tikzpicture}
			
			\caption{The upper bound on $H(\hat{A}_i|\hat{B}_i)$ as a function of the CHSH  winning probability $\omega$ or the violation $\beta$.}\label{fig:single_round_entropy}
			\end{figure}


\section{Preliminaries}\label{sec:prelim}

	\subsection{General notation}
		
		The set $\{1,2,\dotsc,n\}$ is denoted by $[n]$. 
		
		All logarithms are in base 2. 
		Random variables (RV) are denoted by capital letters while specific values are denoted by small letters.
		Sets are denoted with calligraphic fonts.
		For example, we use $X_i$ to denote a RV over $\mathcal{X}$ and $x_i$ to denote a certain value $x_i\in\mathcal{X}$.

		Most RV will refer to a specific $i\in[n]$ denoted in their subscript, as in $X_i$ above. RV describing a range between $i,j\in[n]$ for $i<j$ are written as $X_{i,\cdots,j}=X_i,\cdots,X_j$. When no subscript appears then the range is from $1$ to $n$, i.e., $X=X_1,\cdots,X_n$.
		
		The fidelity of two quantum states is given by $F(\rho,\sigma)=\| \sqrt{\rho} \sqrt{\sigma}\|_1^2$.

	\subsection{The CHSH  inequality and game}\label{sec:pre_chsh}
		
		In a bipartite scenario where each party, Alice and Bob, can perform one of two possible measurements, indexed by $X_i,Y_i\in\{0,1\}$, and with outcomes denoted by $A_i,B_i\in\{0,1\}$, a Bell inequality is a linear constraint on the conditional probability distributions $P(A_i,B_i|X_i,Y_i)$ which is satisfied by all local hidden variable models~\cite{brunner2014bell}. One Bell inequality of special interest is the so-called Clauser-Horne-Shimony-Holt (CHSH) inequality~\cite{clauser1969proposed}, which takes the form
		\begin{equation*}
			\beta = \sum_{a_i,b_i,x_i,y_i} (-1)^{a_i+b_i+x_i y_i} P(a_i,b_i|x_i,y_i) \leq 2 \;.
		\end{equation*}
		
		This inequality admits a maximum quantum violation of $\beta=2\sqrt{2}$. To achieve this maximal violation Alice and Bob can share the maximally entangled state $\ket{\Phi^+} = \left(\ket{00}+\ket{11}\right)/\sqrt{2}$ and measure it with the following measurements: Alice's measurements $X_i=0$ and $X_i=1$ correspond to the Pauli operators~$\sigma_z$ and~$\sigma_x$ respectively and Bob's measurements $Y_i=0$ and $Y_i=1$ to $\left(\sigma_z + \sigma_x\right)/\sqrt{2}$ and $\left(\sigma_z - \sigma_x\right)/\sqrt{2}$ respectively.
		We can therefore restrict our attention to $\beta\in[2,2\sqrt{2}]$. 
		
		The CHSH inequality can be equivalently expressed in terms of a game. The game is defined via the  winning condition 
		$w_{\text{CHSH}}=1$ iff $A_i\oplus B_i = X_i \cdot Y_i $. 
		The optimal quantum strategy described above achieves winning probability $\omega = \frac{2+\sqrt{2}}{4}\approx 0.85$, while the optimal classical strategy achieves a winning probability of $0.75$.
		The relation between the winning probability in the  game and the violation of CHSH inequality is given by $\omega=1/2+\beta/8$ and we have $\omega \in \left[ \frac{3}{4}, \frac{2+\sqrt{2}}{4}\right]$.

	\subsection{Quantum entropies and Markov chains}
	
		\paragraph*{Shannon entropy.}
		The Shannon entropy of a probability distribution $\{p_i\}$ is given by $H(\{p_i\})=-\sum_{p_i}p_i\log(p_i)$.
		For a Bernoulli process with probability $p$ the we use the binary entropy function $h(p)=-p\log (p) -(1-p)\log (1-p)$.

		\paragraph*{von Neumann entropy.}
		The von Neumann entropy $H(\rho)$ of a quantum state $\rho$ is given by $H(\rho)=-\Tr (\rho \log \rho)$. 
		Given a bipartite state $\rho_{AB}\in\mathcal{H}_A\otimes\mathcal{H}_B$ the conditional von Neumann entropy is 
		\[
			H(A|B)_{\rho_{AB}} = H(\rho_{AB}) -H(\rho_{B})\;.
		\]
		When the state on which the entropy is evaluated is clear from the context we drop the subscript and write $H(A|B)$.
	
		\paragraph*{Max-entropy.}
		The explicit definition of the quantum max-entropy will not be of use in the current work. Yet, we state it here for completeness. 
		The conditional smooth max-entropy of a bipartite quantum state $\rho_{AB}$ is given by 
		\[
			H^{\varepsilon}_{\max}(A|B)_{\rho_{AB}} = \log \inf_{\sigma_{AB} \in \mathcal{B}^\varepsilon(\rho_{AB})} \sup_{\tau_B} \| \sigma_{AB}^{\frac{1}{2}}\tau_B^{-\frac{1}{2}}\|_1^2 \;,
		\]
		where  $\mathcal{B}^\varepsilon(\rho_{AB})$ the set of sub-normalised states $\sigma_{AB}$ with $P(\rho_{AB},\sigma_{AB}) \leq \varepsilon$ and $P$ is the purified distance~\cite{tomamichel2010entropyduality}.

		\paragraph*{Markov chains.}
		A tripartite quantum state $\rho_{ABC}$ is said to fulfil the Markov chain condition $A\leftrightarrow B \leftrightarrow C$ if $I(A:C|B) = 0$, where $I(A:C|B)= H(AB) + H(BC) - H(B) - H(ABC)$ is the conditional mutual information. Note that $I(A:C|B) = 0$ if and only if given $B$, $A$ and $C$ are independent.
		
	\subsection{The entropy accumulation theorem}\label{sec:eat_def_app}
	
		The entropy accumulation theorem (EAT)~\cite[Theorem 4.4]{dupuis2016entropy} gives us a way of bounding the amount of smooth min- or max-entropy accumulated during a sequential process fulfilling certain conditions. In contrast to previous works where a bound on the smooth min-entropy was derived using the EAT, the current work uses the EAT to bound the smooth max-entropy. 
		We state here the necessary details in the context of our work. 

		To apply the EAT one needs to define ``EAT channels'' which describe the sequential process under consideration (in our case, for example, the channels are defined via the actions of our DIEC protocol). EAT channels are defined as follows.
		\begin{defn}[EAT channels]\label{def:eat_channels}
			EAT channels $\mathcal{N}_i:R_{i-1}\rightarrow R_i O_i S_i W_i$, for $i\in[n]$, are CPTP maps such that for all $i\in [n]$:
			\begin{enumerate}
				\item $O_i$ are finite dimensional quantum systems of dimension $d_{O_i}$ and $W_i$ are finite-dimensional classical systems (RV).  
					$S_i$ and $R_i$ are arbitrary quantum systems. 
				\item For any $i\in[n]$ and any input state $\sigma_{R_{i-1}}$, the output state $\sigma_{R_i O_i S_iW_i} =  \mathcal{N}_i  \left(\sigma_{R_{i-1}} \right)$ has the property that the classical value $W_i$ can be measured from the marginal $\sigma_{O_iS_i}$ without changing the state. 
						That is, for the map $\mathcal{T}_i :O_iS_i\rightarrow O_iS_iW_i$ describing the process of deriving $W_i$ from $O_i$ and $S_i$, it holds that $\Tr_{W_i}\circ\mathcal{T}_i \left(\sigma_{O_iS_i}\right) = \sigma_{O_iS_i} $.
				\item For any initial state 
				$\tau_{R_0}^0$, the final state $\tau_{OSW} = \left( \Tr_{R_n}\circ \mathcal{N}_n \circ \dots \circ \mathcal{N}_1\right) \; \tau_{R_0}^0$ 
				fulfils the Markov chain condition $O_{1\dotsc i-1} \leftrightarrow S_{1\dotsc i-1} \leftrightarrow S_i$ for each $i\in[n]$. 
			\end{enumerate}
		\end{defn}

		We will use below the following notation. Given a value $w=w_1, \dotsc, w_n \in \mathcal{W}^n$, where $\mathcal{W}$ is a finite alphabet, we denote by $\mathrm{freq}_w$ the probability distribution over $\mathcal{W}$ defined by $\mathrm{freq}_w(\tilde{w}) = \frac{| \left\{ i | w_i = \tilde{w} \right\} |}{n}$ for $\tilde{w}\in\mathcal{W}$.
		If $\tau$ is a state classical on $W$ we write $\Pr\left[w\right]_{\tau}$ to denote the probability that $\tau$ assigns to $w$. 
		
		\begin{defn}[Max-tradeoff functions]\label{def:max_tradeoff_func}
			Let $\mathcal{N}_1,\ldots,\mathcal{N}_n$ be a family of EAT channels. Let $\mathcal{W}$ denote the common alphabet of $W_1,\ldots,W_n$. 	
			A concave\footnote{Let  $\hat{\Omega}$ be a set of frequencies defined via $\mathrm{freq}_w(\tilde{w})\in\hat{\Omega}$ if and only if $\tilde{w}\in\Omega$.
			We can consider concave functions, in contrast to affine ones~\cite{dupuis2016entropy}, since the event $\Omega$ defined in the current work results in a \emph{convex set} $\hat{\Omega}$.}  function $f_{\max}$ from the set of probability distributions~$p$ over $\mathcal{W}$ to the real numbers is called a \emph{max-tradeoff function} for $\{\mathcal{N}_i\}$ if it satisfies
			\[
				f_{\max}(p) \geq \sup_{\sigma_{R_{i-1}}:\mathcal{N}_i(\sigma)_{W_i}=p} H\left( O_i | S_i  \right)_{\mathcal{N}_i(\sigma)} \;
			\]
			for all $i\in [n]$, where the supremum is taken over all input states of $\mathcal{N}_i$ for which the marginal on $W_i$ of the output state is the probability distribution $p$.  
		\end{defn}

		The statement of the EAT, relevant for the smooth max-entropy, is given below.
		\begin{thm}[EAT]\label{thm:eat}
			Let $\mathcal{N}_i:R_{i-1}\rightarrow R_i O_i S_i W_i$ for $i\in [n]$ be EAT channels as in Definition~\ref{def:eat_channels},  $\tau_{OSW} = \left( \Tr_{R_n}\circ \mathcal{N}_n \circ \dots \circ \mathcal{N}_1\right) \tau_{R_0}$ be the final state, 
			$\Omega$ an event  defined over $\mathcal{W}^n$, $\Pr[\Omega]_{\tau}$ the probability of $\Omega$ in $\tau$, 
			and $\tau_{|\Omega}$ the final state conditioned on $\Omega$. Let $\varepsilon_{\mathrm{smo}} \in (0,1)$.
			
			For $f_{\max}$ a max-tradeoff function for  $\{\mathcal{N}_i\}$ as in Definition~\ref{def:max_tradeoff_func} and any $t\in \mathbb{R}$ such that $f_{\max}\left(\mathrm{freq}_w \right) \leq t$ for any $w\in\mathcal{W}^n$ for which $\Pr\left[w\right]_{\tau_{|\Omega}}> 0$,
			\[
				H_{\max}^{\varepsilon_{\mathrm{smo}}} \left( O|S \right)_{\tau_{|\Omega}} \leq n t + v\sqrt{n} \;,
			\]
			where $v = 2\left(\log(1+2 d_{O_i} ) + \lceil \|  \triangledown f_{\max} \|_\infty \rceil \right)\sqrt{1-2\log (\varepsilon_{\mathrm{smo}} \cdot \Pr[\Omega]_{\tau})}$ and $d_{O_i}$ denotes the dimension of~$O_i$.
		\end{thm}
		
		We used above a slightly and trivially modified statement of the EAT, compared to that of~\cite{dupuis2016entropy}. 
		The definition of the max-tradeoff function used in~\cite{dupuis2016entropy} consideres $H\left( O_i | S_i  R' \right)$ where $R'$ is isomorphic to to $R_{i-1}$. For the calculation of the supremum one can always assume that the system on $R'$ is in product with the rest of the system and hence drop it here (see Remark~4.2 in~\cite{dupuis2016entropy}). 
		Furthermore, in~\cite{dupuis2016entropy} the EAT is stated with $H_{\max}^{\varepsilon_{\mathrm{smo}}}(O|SE)$ where $E$ denotes a system extending the initial state $\tau_{R_0}$, i.e., we have $\Tr_E(\tau_{R_0E}) = \tau_{R_0}$. 
			The $E$, in fact, can be dropped as the theorem must hold for any $\tau_{R_0E}$ and, hence, in particular for $\tau_{R_0E}=\tau_{R_0}\otimes\tau_E$ as in Theorem~\ref{thm:eat}, for which the conditional smooth max-entropy is maximal.
			In our context, an alternative way of thinking about this is to note that our goal is to bound $H_{\max}^{\varepsilon_{\mathrm{smo}}} \left( O|S \right)_{\tau_{|\Omega}}$. As it clearly depends only on the registers~$O$ and~$S$, $E$ does not take part in the proof and we can, w.l.o.g.\@, consider an initial state of the form $\tau_{R_0}\otimes\tau_E$. The final state then also has a tensor product form.


\section{Main parts of the proof}\label{sec:main_proofs}

	In this section we present the main steps and ideas used in the proof of Theorem~\ref{thm:main_thm}. The full details are given in the Appendix. 

	\subsection{Modified protocol}\label{sec:modified protocol}
	
		As explained above, our goal is to show that there exists an entanglement distillation protocol that can distill the entanglement present in  $\rho_{|\Omega}$. 
		Instead of considering an entanglement distillation protocol that acts on the state $\rho_{|\Omega}$ directly, we consider a slightly modified scenario. The modified scenario will result in a state $\tau_{|\Omega}$, closely related to $\rho_{|\Omega}$, from which at least same amount of entanglement can be distilled. 
		
		Concretely, we modify the real DIEC, Protocol~\ref{pro:ent_test}, to define the modified protocol stated as Protocol~\ref{pro:mod_ent_test}. The only difference between this protocol and Protocol~\ref{pro:ent_test} is in Steps~\ref{prostep:state_after_proj} and~\ref{prostep:symm_maps}, which we explain below.
		
		We remark that Protocol~\ref{pro:mod_ent_test} is being used only as a step in the proof of our theorem. We do not claim at any point that this modified protocol can be implemented by the verifier given the uncharacterised devices (in fact, it cannot). 
		It will become clear later on why the modified protocol is, nevertheless, needed in our proof.
		
		\begin{algorithm}
			\caption{Modified DIEC (based on the CHSH inequality)}
			\label{pro:mod_ent_test}
			\begin{algorithmic}[1]
				\STATEx \textbf{Arguments:} 
					\STATEx\hspace{\algorithmicindent} $D$ -- untrusted measurement device of two components with inputs and outputs set $\{0,1\}$
					\STATEx\hspace{\algorithmicindent} $n \in \mathbb{N}_+$ -- number of rounds
					\STATEx\hspace{\algorithmicindent} $\omega_{\mathrm{exp}}$ -- expected winning probability in the CHSH game  
					\STATEx\hspace{\algorithmicindent} $\delta_{\mathrm{est}} \in (0,1)$ -- width of the statistical confidence interval for the estimation test
					
				\STATEx
				
				\STATE For every round $i\in[n]$ do Steps~\ref{prostep:first}-\ref{prostep:symm_maps}:
					\STATE\hspace{\algorithmicindent}Let $\phi^i$ denote the bipartite state produced by the source in this round. \label{prostep:first}
					
					\STATE\hspace{\algorithmicindent}Set $A_i,B_i,C_i,D_i,X_i,Y_i,W_i=\perp$.
					\STATE\hspace{\algorithmicindent}Choose $T_i=1$ with probability $\gamma$ and $T_i=0$ otherwise.
					\STATE\hspace{\algorithmicindent}If $T_i=1$: 
						\STATE\hspace{\algorithmicindent}\hspace{\algorithmicindent}Choose the inputs $X_i,Y_i\in\{0,1\}$ uniformly at random. 
						\STATE\hspace{\algorithmicindent}\hspace{\algorithmicindent}Measure $\phi^i$ using $D$ with the inputs $X_i,Y_i$. Record the outputs as $A_i,B_i\in\{0,1\}$. 
						\STATE\hspace{\algorithmicindent}\hspace{\algorithmicindent}Set $W_i = 1$ if the CHSH game is won and $W_i=0$ otherwise.
					\STATE\hspace{\algorithmicindent}If $T_i=0$:
						
						\STATE\hspace{\algorithmicindent}\hspace{\algorithmicindent}Keep $\phi^i$ in the registers $\hat{A}_i\hat{B}_i$. 
						
						\STATE\hspace{\algorithmicindent}\hspace{\algorithmicindent}Apply the projection as in Equation~\eqref{eq:state_after_proj} to create $\bar{\rho}^i_{\hat{A}_i\hat{B}_iC_iD_i}$. \label{prostep:state_after_proj}
						\STATE\hspace{\algorithmicindent}\hspace{\algorithmicindent}\parbox[t]{\dimexpr\linewidth-\algorithmicindent}{Choose a random unitary $U$ uniformly within the set $\{\mathbb{I},\sigma_x,\sigma_y,\sigma_z\}$ and apply it to create $\tilde{\rho}^i_{\hat{A}_i\hat{B}_iC_iD_i}$ as in Equation~\eqref{eq:symm_maps}.\strut} \label{prostep:symm_maps}

					\STATE Abort if $\sum_i \chi(T_i=1)\,W_i < (\omega_{\mathrm{exp}}\gamma - \delta_{\mathrm{est}}) n \;$. \label{prostep:abort_ET}
			\end{algorithmic}
			\end{algorithm}

			\subsubsection{First modification: reduction to qubits}
			
				When a test is not being performed by the verifier the state produced by the source is being kept in the memory. In Protocol~\ref{pro:ent_test} the state is being kept as is. In contrast, in Protocol~\ref{pro:mod_ent_test} we first project the state such that the resulting state is a two qubit state. The projection, described by two local projections, can obviously only decrease the amount of entanglement of the kept state. 
				
				As one can guess, we define the specific projection via the measurements performed by the measurement devices.\footnote{These are unknown to the verifier but, as mentioned above, Protocol~\ref{pro:mod_ent_test} does not need to be implemented in practice by the verifier. The only thing that matters is that such projectors \emph{exist}.}
				Formally, due to Jordan's lemma (see, e.g.,~\cite[Lemma 4.1]{scarani2012device}), the measurements operators $\Pi_{a|x}^{A_i}$, $\Pi_{b|y}^{B_i}$ of each party, Alice and Bob, in each round $i\in[n]$ can be written in a suitable local basis as
				\begin{equation*}\label{eq:settings}
				\begin{split}
					\Pi^{A_i}_{0|0} - \Pi^{A_i}_{1|0} &= \underset{c_i}\oplus A_0^{c_i} = \underset{c_i}\oplus \sigma_z^{c_i}\\
					\Pi^{B_i}_{0|0} - \Pi^{B_i}_{1|0} &= \underset{d_i}\oplus B_0^{d_i} = \underset{d_i}\oplus \sigma_z^{d_i}\\
					\Pi^{A_i}_{0|1} - \Pi^{A_i}_{1|1} &= \underset{c_i}\oplus A_1^{c_i}= \underset{c_i}\oplus \cos(\alpha_{c_i}) \sigma_z^{c_i} + \sin(\alpha_{c_i})\sigma_x^{c_i}\\
					\Pi^{B_i}_{0|1} - \Pi^{B_i}_{1|1} &= \underset{d_i}\oplus B_1^{d_i} = \underset{d_i}\oplus \cos(\beta_{d_i}) \sigma_z^{d_i} + \sin(\beta_{d_i})\sigma_x^{d_i}
				\end{split}
				\end{equation*}
				with parameters $\alpha_{c_i}$ and $\beta_{d_i}$ and the different $\sigma$'s denoting the Pauli matrices.

				With this in mind, when $T_i=0$ we let each party perform a measurement on $\rho^i_{\hat{A}_i\hat{B}_i}$ (which is identical to $\phi^i$ in the case $T_i=0$) according to the following operators:
				\begin{equation}\label{eq:jordan_proj}
				\begin{split}
					\Pi^{\hat{A}_i}_{c_i} &= \ketbra{c_i}{c_i}\otimes \mathbb{I}\\
					\Pi^{\hat{B}_i}_{d_i} &= \ketbra{d_i}{d_i}\otimes \mathbb{I}.
				\end{split}
				\end{equation}
				This gives  each party the index of her/his respective Jordan blocks $c_i$ and $d_i$. 
				We denote by $p(c_i,d_i)$ the probability of observing the indices $c_i,d_i$. 
				The global system is projected into the two-qubit state 
				\begin{equation}\label{eq:state_after_proj}
					\bar{\rho}^i_{\hat{A}_i\hat{B}_i C_i D_i} = \frac{\left(\Pi^{\hat{A}_i}_{c_i} \otimes \Pi^{\hat{B}_i}_{d_i}\right)\rho^i_{\hat{A}_i\hat{B}_i} \left(\Pi^{\hat{A}_i}_{c_i} \otimes \Pi^{\hat{B}_i}_{d_i}\right)}{\Tr\left(\Pi^{\hat{A}_i}_{c_i} \otimes \Pi^{\hat{B}_i}_{d_i}\rho^i_{\hat{A}_i\hat{B}_i}\right)}\otimes \ketbra{c_i}{c_i}\otimes\ketbra{d_i}{d_i}\;,
				\end{equation}
				where we added a classical register for the measurement outcomes. At this stage, the parties exchange their indices $c_i$ and $d_i$ with each other. It is easy to see that the considered projection step can be described via LOCC.

				Note the following:
				\begin{enumerate}
					\item Due to the projection, the registers $\hat{A}_i\hat{B}_i$ are now qubit registers. 
					\item Since the projectors given in Equation~\eqref{eq:jordan_proj} are defined via the measurements used by the device when $T_i=1$, they could also be applied when $T_i=1$ without changing the resulting state. This is made formal in Lemma~\ref{lem:proj_idn_subspace} below (for the proof see Appendix~\ref{sec:proofs_mod_prot}). 
				\end{enumerate}

				\begin{lemma}\label{lem:proj_idn_subspace}
					Consider a scenario in which the projection to the two qubit space is applied on the state $\phi^i$ directly after it is produced by the source in the $i$'th round (i.e., before choosing the value of $T_i$). Denote the resulting state in the end of the $i$'th round in such a case by $\bar{\bar{\rho}}^i$. Then 
					\[
						\bar{\bar{\rho}}^i_{\hat{A}_i\hat{B}_i A_iB_i C_i D_iX_iY_i} = \bar{\rho}^i_{\hat{A}_i\hat{B}_i A_iB_i C_i D_iX_iY_i}  \;,
					\]
					where $\bar{\rho}^i$ is as defined in Equation~\eqref{eq:state_after_proj}.
				\end{lemma}
	
				It is clear that for the state $\bar{\bar{\rho}}^i$ the measurements in the case of $T_i=1$ act on the subspace of the same Hilbert space which is otherwise kept in $\hat{A}_i\hat{B}_i$ when $T_i=0$, since the projection is made before choosing the value of $T_i$. Lemma~\ref{lem:proj_idn_subspace} implies that this is also the case when the projection is applied only when $T_i=0$ as in our modification.
				Without this property one could not argue that the test rounds, i.e., those for which $T_i=1$, represent also the rounds with $T_i=0$.

			\subsubsection{Second modification: reduction to Bell diagonal states}
			
				After the projection, we apply another step that can be seen as a symmetrisation step, also called twirling in the literature.	
				In each round for which $T_i=0$, a one-qubit unitary $U$ is chosen uniformly within the set $\{\mathbb{I},\sigma_x,\sigma_y,\sigma_z\}$ and applied on both systems $\hat{A}$ and $\hat{B}$. That is,
				\begin{equation}\label{eq:symm_maps}
					\tilde{\rho}^i_{\hat{A}_i\hat{B}_i|c_i,d_i} = \frac{1}{4} \sum_{U} \left(U\otimes U\right) \bar{\rho}^i_{\hat{A}_i\hat{B}_i|c_i,d_i} \left(U\otimes U\right)^{\dagger} \;,
				\end{equation}
				where $\bar{\rho}^i_{\hat{A}_i\hat{B}_i}$ is defined in Equation~\eqref{eq:state_after_proj} and the sum is over $U\in\{\mathbb{I},\sigma_x,\sigma_y,\sigma_z\}$.
				
				It was shown in~\cite{bennett1996mixed} that the resulting state is diagonal in the Bell basis $\{\ket{\Phi^+},\ket{\Phi^-},\ket{\Psi^+},\ket{\Psi^-}\}$. In our notation, we get the following corollary:
				
				\begin{cor}\label{cor:bell_diagonal}
					For all $i\in[n]$ with $T_i=0$,  
					\begin{equation}\label{eq:mixt_of_bell_mix}
						\tilde{\rho}^i_{\hat{A}_i\hat{B}_iC_iD_i}=\sum_{c_i,d_i} p(c_i,d_i)\;\tilde{\rho}^{i}_{\hat{A}_i\hat{B}_i|c_i,d_i}\otimes \ketbra{c_i}{c_i}\otimes\ketbra{d_i}{d_i}
					\end{equation} 
					where each $\tilde{\rho}^{i}_{\hat{A}_i\hat{B}_i|c_i,d_i}$ is diagonal in the Bell basis. That is, it can be written as 
					\begin{equation}\label{eq:bell_mixt}
						\tilde{\rho}^{i}_{\hat{A}_i\hat{B}_i|c_i,d_i} = \begin{pmatrix} \lambda_{\Phi^+} & & & \\  & \lambda_{\Phi^-} & &  \\  &  & \lambda_{\Psi^+} & \\  &  & & \lambda_{\Psi^-}   \end{pmatrix}
					\end{equation}
					in the basis of the Bell states ordered as $\{\ket{\Phi^+},\ket{\Phi^-},\ket{\Psi^+},\ket{\Psi^-}\}$ for some eigenvalues. 			
				\end{cor}

				As in the projection step, the effect of applying a random unitary is restricted to the case $T_i=0$ and can be done with only LOCC. The importance of this step lies in the following:
				\begin{enumerate}
					\item After the twirling step, the resulting two qubit state is a convex combination of Bell diagonal states (Corollary~\ref{cor:bell_diagonal}). This simplifies the analysis in the rest of the proof.
					\item This step decouples $\tilde{\rho}^i_{\hat{A}_i\hat{B}_iC_iD_i}$ from other registers while maintaining the entanglement of $\tilde{\rho}^i_{\hat{A}_i\hat{B}_i}$; see Lemma~\ref{lem:symm_dec} below for the formal statement. This property is crucial later on, in Section~\ref{sec:prereq_eat}. 
				\end{enumerate}

			\subsubsection{Properties of the modified protocol}
			
			Protocol~\ref{pro:mod_ent_test} can be described mathematically by an application of a sequence of maps one after the other. The maps describe both the actions of the verifier defined by the protocol as well as the actions of the measurement devices.  We denote these maps by
			\begin{equation}\label{eq:maps_mod_prot_def}
				\mathcal{N}_i:R_{i-1}\rightarrow R_i\hat{A}_i\hat{B}_iA_iB_iC_iD_iX_iY_iW_i \;.
			\end{equation}
			The register $R_{i-1}$ holds here the state of the uncharacterised devices in the beginning of the $i$'th round. That is, $R_{i-1}$ includes the state $\phi^i$ produced by the source as well as any other information kept by the measurement devices. $R_i$ describes a register which can be used as internal memory for the devices.  
			
			The final state in the end of Protocol~\ref{pro:mod_ent_test} is denoted by 
			\[
				\tau = \tau_{\hat{A}\hat{B}ABCDXYTW} 
			\]
			and the state conditioned on not aborting by $\tau_{|\Omega}$.
			
			We now present 3 statements regarding the modified protocol. All the proofs are  rather simple and are given in Appendix~\ref{sec:proofs_mod_prot}.
		
			For start, it follows from the definition of Protocol~\ref{pro:mod_ent_test} that the observed statistic of $\rho$ (the state in the end of Protocol~\ref{pro:ent_test}) and $\tau$ (the state in the end of Protocol~\ref{pro:mod_ent_test}) are the same. This, in particular, implies that the probabilities of aborting Protocol~\ref{pro:ent_test} and Protocol~\ref{pro:mod_ent_test} are identical. 
			\begin{lemma}\label{lem:mod_same_stat}
				The observed statistics and, hence, the probabilities of aborting Protocol~\ref{pro:ent_test} and Protocol~\ref{pro:mod_ent_test} are the same. That is,
				$\rho_{ABXYTW} = \tau_{ABXYTW}$ and
				$\Pr[\neg\Omega]_{\rho} = \Pr[\neg\Omega]_{\tau}$.
			\end{lemma}
			
			Second, as mentioned in the previous section, the final state of Protocol~\ref{pro:mod_ent_test}, $\tau$, has the property that for any $i\in[n]$ the registers $\hat{A}_i \hat{B}_i$ are decoupled from the other registers. This can be shown using the structure of the state given in Equation~\eqref{eq:mixt_of_bell_mix}. 
			The formal statement needed in the next sections is given in the next lemma.
		
			\begin{lemma}\label{lem:symm_dec}
				Let $\tau$ denote the state after all rounds of Protocol~\ref{pro:mod_ent_test} (before conditioning on $\Omega$).
				For every~$i\in[n]$ let $K$ include all registers different than $\hat{A}_i \hat{B}_iC_iD_iT_i$ such that we can write 
				\[
					\tau_{\hat{A}\hat{B}ABCDXYTW} = \tau_{\hat{A}_i \hat{B}_iC_iD_iT_iK} \;.
				\] 
				Then,
				\begin{equation*}
					H(\hat{A}_i|\hat{B}_iC_iD_iT_iK)_{\tau} = H(\hat{A}_i|\hat{B}_iC_iD_iT_i)_{\tau} \;.
				\end{equation*}
			\end{lemma}

			The last thing which will be of use later is that $\rho_{|\Omega}$ is at least as entangled as $\tau_{|\Omega}$. That is, if one can distill $\ket{\Phi^L}$ from $\tau_{|\Omega}$ then one can also distill $\ket{\Phi^L}$, for the same $L$, from $\rho_{|\Omega}$ 
			This follows from the fact that both modifications described above can be implemented using LOCC, and hence they cannot increase the entanglement. The formal statement is given in the last lemma of this section. 
	
			\begin{lemma}\label{lem:locc_tau_to_rho}
				Let $\Gamma$ denote the LOCC protocol used to distill 
				$\ket{\Phi^L}= \frac{1}{\sqrt{L}}\sum_{i=1}^L \ket{i}\ket{i}$
				from $\tau_{|\Omega}$ with error probability $\varepsilon$. 
				Then, there exists another LOCC protocol $\Delta$ which can be used to distill $\ket{\Phi^L}$ from $\rho_{|\Omega}$ with the same error probability.
			\end{lemma}

			Lemma~\ref{lem:locc_tau_to_rho} implies that instead of proving that there exists an entanglement distillation protocol starting from $\rho_{|\Omega}$ one can prove that there exists an entanglement distillation protocol starting from~$\tau_{|\Omega}$. The advantage of doing this is that $\tau_{|\Omega}$ has some nice properties as was shown above. These will be of use in our proof given in the next sections. 
			We emphasise again that Protocol~\ref{pro:mod_ent_test} acts just as a step in the proof. This is not a real protocol which we expect the verifier to implement in order to verify that the state produced by the source is highly entangled. 
	
	
	\subsection{Single-round bound on the von Neumann entropy}\label{sec:single_round_ent}
	
		The goal of the current section is to upper bound the conditional von Neumann entropy $H(\hat{A}_i|\hat{B}_i)$ of the states produced by the source and kept in the memory in Protocol~\ref{pro:mod_ent_test}. As explained in Section~\ref{sec:von_neumann_intro}, a negative value of the conditional von Neumann entropy can be used to quantify the amount of entanglement present in the considered state. While there are many ways to quantify  entanglement, the conditional von Neumann entropy is the one relevant for our proof technique (as will become clear below). 
			
		The first lemma of this section was already presented as Lemma~\ref{lem:Bell_diag_entropy_intro} in Section~\ref{sec:von_neumann_intro}.
		To simplify the equations in the lemmas below we use the notation of a Bell violation $\beta\in\left[2,2\sqrt{2}\right]$ instead of the winning probability $\omega\in\left[\frac{3}{4},\frac{2+\sqrt{2}}{4}\right]$ (used in Lemma~\ref{lem:Bell_diag_entropy_intro}). The two are related via $\omega=\frac{1}{2}+\frac{\beta}{8}$; the transformation back to $\omega$ will be done in the end.
			
		\begin{lemma}\label{lem:single_ent_diagonal_Bell}
			For any Bell diagonal state $\sigma_{\hat{A}_i\hat{B}_i}$ as in Equation~\eqref{eq:bell_mixt}  that can be used to violate the CHSH inequality with violation $\beta\in\left[2,2\sqrt{2}\right]$,
			\begin{equation*}
				H(\hat{A}_i|\hat{B}_i)_{\sigma} \leq \mathscr{H}(\beta) - 1 \;,
			\end{equation*}
			where
			\begin{equation}\label{eq:H_of_beta}
				\mathscr{H}(\beta) = 2h\left(\frac{1}{2}-\frac{\beta}{4\sqrt{2}}\right) 
			\end{equation}
			and $h$ is the binary entropy function.
		\end{lemma}
	
		The function $\mathscr{H}(\beta) - 1$ is plotted in Figure~\ref{fig:single_round_entropy} as a function of the CHSH violation $\beta$ and the winning probability $\omega$.
	
		The proof is inspired by the work of~\cite{pironio2009device}, though we derive a bound on a different quantity. Neither result follows from the other.
		We give the proof sketch here; the full details are given in Appendix~\ref{sec:proofs_sing_ent}.
		
		\begin{proof}[Proof sketch]
			First note that 
			\begin{equation*}\label{eq:cond_ent_to_comb}
				H(\hat{A}_i|\hat{B}_i)  = H(\hat{A}_i\hat{B}_i) - H(\hat{B}_i) = H(\hat{A}_i\hat{B}_i) -1 
			\end{equation*}
			since the marginal $\sigma_{\hat{B}_i}$ is a completely mixed qubit state (for a Bell diagonal state $\sigma_{\hat{A}_i\hat{B}_i}$). 
			
			We are left to upper bound $H(\hat{A}_i\hat{B}_i)$ which, given that the state is Bell diagonal, is simply the Shannon entropy $H(\vec{\lambda})$ of the probability distribution $\vec{\lambda}=(\lambda_{\Phi^+},\lambda_{\Phi^-},\lambda_{\Psi^+},\lambda_{\Psi^-})$ defined by the four eigenvalues of $\sigma$.\footnote{This observation was already made in~\cite{bennett1996concentrating}. In the context of~\cite{bennett1996concentrating}, the eigenvalues $\vec{\lambda}$ are known. What we do here can be seen as an extension to the case where only the Bell violation is known and not the eigenvalues.} Hence, our goal is to maximise $H(\vec{\lambda})$ under the constraint of having the correct Bell violation $\beta$.  
			
			Following~\cite[Lemma 7]{pironio2009device} and~\cite[Equation (21)]{horodecki1995violating} we translate the constraint on the violation of the state to the following constraint on the eigenvalues:
			\begin{equation*}
			\begin{split}
				\beta = \max \Big\{ 
							& 2\sqrt{2} \sqrt{(\lambda_{\Phi^+}-\lambda_{\Psi^+})^2+(\lambda_{\Phi^-}-\lambda_{\Psi^-})^2} \;, \\
							& 2\sqrt{2} \sqrt{(\lambda_{\Phi^+}-\lambda_{\Psi^-})^2+(\lambda_{\Phi^-}-\lambda_{\Psi^+})^2} \;, \\ 
							& 2\sqrt{2} \sqrt{(\lambda_{\Phi^+}-\lambda_{\Phi^-})^2+(\lambda_{\Psi^+}-\lambda_{\Psi^-})^2} 
						\Big\} \;.
			\end{split}
			\end{equation*}
	
			Using the symmetry of our problem (e.g., under the exchange $\lambda_{\Psi^+} \leftrightarrow \lambda_{\Psi^-}$) one can simplify the above constraint and write the optimisation problem of interest as
			\begin{equation}\label{eq:optimisation_prob_main_txt}
			\begin{split}
				\max \quad &H(\vec{\lambda}) \\
				\text{s.t.} \quad 
				&\lambda_{\Phi^+} = \frac{1}{2}\left(1- \lambda_{\Psi^-} - \lambda_{\Phi^-} + \sqrt{\frac{\beta^2}{8} - \left(\lambda_{\Psi^-} - \lambda_{\Phi^-}\right)^2}\right) \\
				& \lambda_{\Phi^+},\lambda_{\Psi^+},\lambda_{\Phi^-},\lambda_{\Psi^-}\geq 0 \\
				& \lambda_{\Phi^+}+\lambda_{\Psi^+}+\lambda_{\Phi^-}+\lambda_{\Psi^-} =1
			\end{split}
			\end{equation}
			
			The constraints of the optimisation problem imply that we can write $H(\vec{\lambda})$ as a function of only two variables, $\lambda_{\Psi^-}$ and $\lambda_{\Phi^-}$, for any value of~$\beta$. As an example, $H(\vec{\lambda})$ is presented in Figure~\ref{fig:H_lambda} for $\beta=2.5$.

			One can solve this optimisation problem numerically; the solution is given by  
			\begin{equation}\label{eq:opt_sol_main_txt}
			\begin{split}
				&\lambda_{\Phi^+}^*=  \left(\frac{1}{2}-\frac{\beta}{4\sqrt{2}}\right)^2 \;; \quad
				\lambda_{\Psi^+}^* = \left(\frac{1}{2}+\frac{\beta}{4\sqrt{2}}\right)^2 \;; \\
				&\lambda_{\Phi^-}^* = \lambda_{\Psi^-}^* = \left(\frac{1}{2}-\frac{\beta}{4\sqrt{2}}\right)\left(\frac{1}{2}+\frac{\beta}{4\sqrt{2}}\right)  \;.
			\end{split}
			\end{equation}
			
			\begin{figure}
				\centering
				\includegraphics[width=0.5\textwidth]{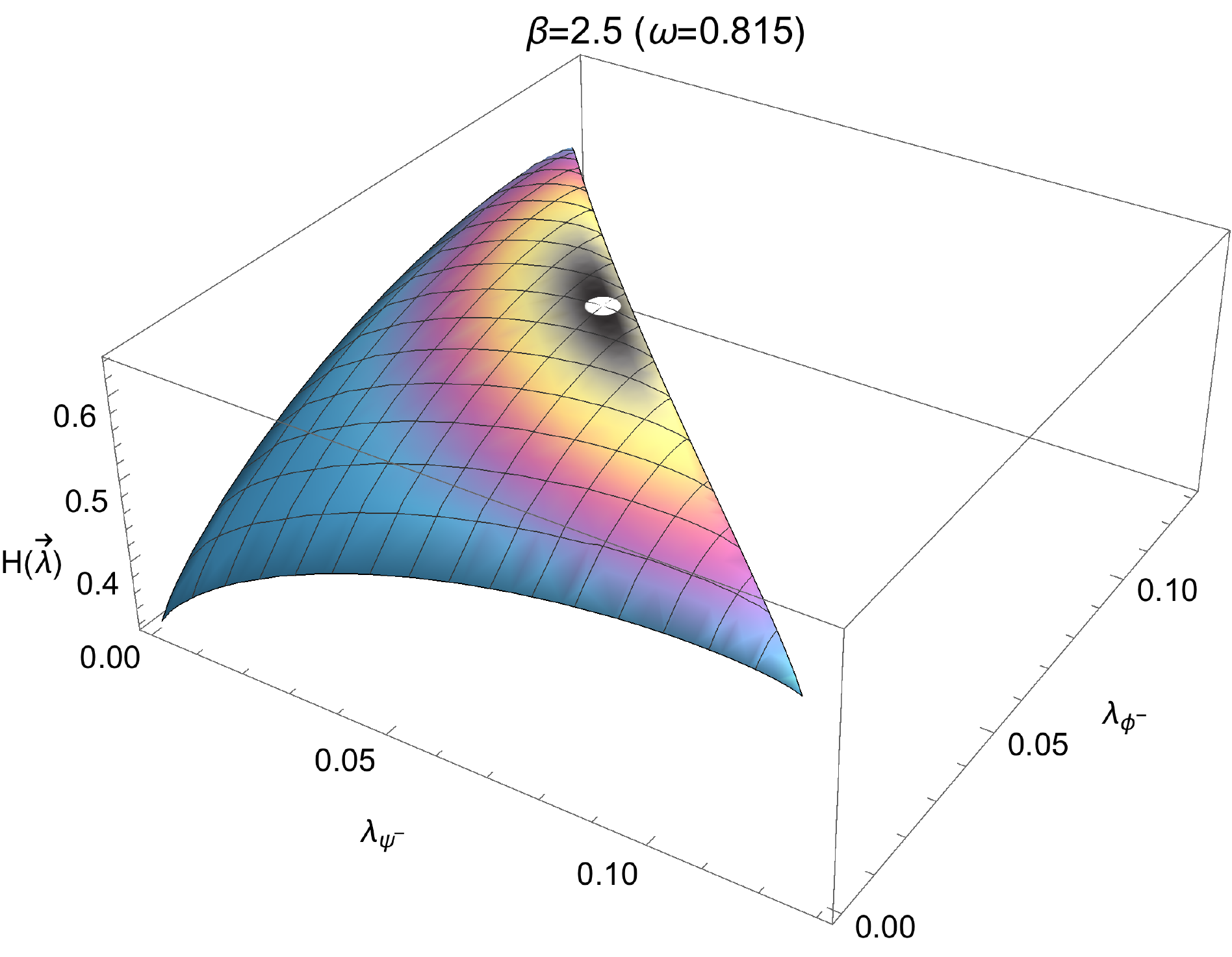}
				\caption{$H(\vec{\lambda})$ as a function of $\lambda_{\Psi^-}$ and $\lambda_{\Phi^-}$, for $\beta=2.5$, in the region defined by the constraints of the optimisation problem stated in Equation~\eqref{eq:optimisation_prob2}. As clearly seen in the plot, there is only a single maxima within the considered region.
				The white point denotes the solution given in Equation~\eqref{eq:opt_sol_main_txt}.}\label{fig:H_lambda}
			\end{figure}
			
			Now that we have a solution in hand, we can verify that it is indeed the correct local optimal solution. This can be done by taking the derivatives in the relevant directions and verifying that the point given in Equation~\eqref{eq:opt_sol_main_txt} is indeed a maxima.
		\end{proof}
			
		Next, we extend the claim of Lemma~\ref{lem:single_ent_diagonal_Bell} also to convex combinations of Bell diagonal states. This can be done easily using the definition of the conditional entropy. The proof is given in Appendix~\ref{sec:proofs_sing_ent}.
			
		\begin{lemma}\label{lem:single_ent_convex_comb_Bell}
			Let $\mathscr{H}(\beta)$ be as in Equation~\eqref{eq:H_of_beta}. 
			Then, for any state $\sigma_{\hat{A}_i\hat{B}_iC_iD_i}$ as in Equation~\eqref{eq:mixt_of_bell_mix} that can be used to violate the CHSH inequality with violation $\beta\in\left[2,2\sqrt{2}\right]$,
			\begin{equation*}
				H(\hat{A}_i|\hat{B}_iC_iD_i)_{\sigma} \leq \mathscr{H}(\beta) - 1 \;.
			\end{equation*}
		\end{lemma}
	
		Two remarks are in order:
		\begin{enumerate}
			\item The derived upper bound on $H(\hat{A}_i|\hat{B}_iC_iD_i)$ is tight; it is saturated by the Bell diagonal state defined via the eigenvalues given in Equation~\eqref{eq:opt_sol_main_txt}. 
			\item As can be seen from Figure~\ref{fig:single_round_entropy}, there is a regime of parameters  ($\beta\lessapprox 2.2$) in which the conditional entropy is positive even though the state violates the CHSH inequality and, hence, is entangled. Indeed, it is known that some states, e.g., the Werner state, can be used to violate the CHSH inequality while presenting positive conditional entropy~\cite{friis2017geometry}. This implies that the conditional entropy is not the optimal quantity to use when certifying entanglement in a DI manner. Yet, it is the relevant quantity when bounding the operationally distillable entanglement (using the known techniques) as we do below.
		\end{enumerate}
	
		For the coming steps of our proof, we need an upper bound on
		\begin{equation*}
			\sup_{\sigma \in \Sigma} H(\hat{A}_i|\hat{B}_iA_iB_iC_iD_iX_iY_i)_{\mathcal{N}_i(\sigma)} \;,
		\end{equation*}
		where $\mathcal{N}_i$ describes the maps defining the rounds of Protocol~\ref{pro:mod_ent_test}, as in Equation~\eqref{eq:maps_mod_prot_def}, and
		$\Sigma$ describes the set of possible states for which Protocol~\ref{pro:mod_ent_test} \emph{does not abort}. 
		On the conceptual level, the protocol does not abort when the observed frequency defined by the registers $W$ imply that the average Bell violation is sufficiently high.\footnote{On the technical level we will need to consider some further details; see Section~\ref{sec:max_tradeoff_func}.} Thus, we shall now consider the set $\Sigma$ defined as $\Sigma = \{ \sigma | w\left(\sigma\right) \geq \omega_{\mathrm{th}}\}$ where $w\left(\sigma\right)$ is the winning probability in the CHSH game of the state $\sigma$ and $\omega_{\mathrm{th}}$ is some threshold winning probability. 
		
		The following lemma can be proven using the above lemmas together with the definition of the conditional entropy and the transformation between $\beta$ and $\omega$; see Appendix~\ref{sec:proofs_sing_ent} for the proof. 
		\begin{lemma}\label{lem:final_single_round_ent}
			For any $\omega_{\mathrm{th}}\in\left[\frac{3}{4},\frac{2+\sqrt{2}}{4}\right]$, let $\Sigma = \{ \sigma | w\left(\sigma\right) \geq \omega_{\mathrm{th}}  \}$ . Then,
			\begin{equation*}
				\sup_{\sigma \in \Sigma} H(\hat{A}_i|\hat{B}_iA_iB_iC_iD_iX_iY_i)_{\mathcal{N}_i(\sigma)} \leq \left(1-\gamma \right) \cdot g(\omega_{\mathrm{th}}) \;,
			\end{equation*}
			where 
			\[
				g(\omega_{\mathrm{th}})=2h\left(\frac{1}{2}-\frac{2\omega_{\mathrm{th}} -1}{\sqrt{2}}\right) -1
			\]
			 and	$h$ is the binary entropy function.
		\end{lemma}

	
	\subsection{Upper bound on the total smooth max-entropy}\label{sec:upper_bound_max_ent}
	
		As mentioned in Section~\ref{sec:diec_result}, a lower bound on the one-shot distillable entanglement can be calculated using the conditional smooth max-entropy $H_{\max}^{\varepsilon_{\mathrm{smo}}}$. 
		The following quantitive relation was derived in~\cite[Proposition 21]{wilde2016converse} :
		\begin{lemma}[\cite{wilde2016converse}]\label{lem:ent_dist}
			For any $\theta_{\hat{A}\hat{B}}$, $\varepsilon_{\mathrm{dist}}\in\left[ 0,1\right]  $, and $\varepsilon_{\mathrm{smo}}\in \lbrack 0,\sqrt{\varepsilon_{\mathrm{dist}}})$ there exists a one-way entanglement distillation protocol $\Gamma$, utilising classical communication from Alice to Bob, such that%
			\[
				F(\Phi^L,\Gamma(\theta_{\hat{A}\hat{B}}))\geq 1-\varepsilon_{\mathrm{dist}},
			\]
			where $\Phi^L$ is a maximally entangled state of Schmidt rank $L$ for
			\[
				\log L=-H_{\max}^{\varepsilon_{\mathrm{smo}}}(\hat{A}|\hat{B})_{\theta}-4\log\left(  \frac{1}{\sqrt{\varepsilon_{\mathrm{dist}}}-\varepsilon_{\mathrm{smo}}}\right)  \;.
			\]
		\end{lemma}
	
		The objective of this section is, thus, to supply a negative upper bound on 
		\[
			H^{\varepsilon_{\mathrm{smo}}}_{\max}(\hat{A}|\hat{B}ABCDXY)_{\tau_{|\Omega}} 	\;.
		\]

		We do this using the entropy accumulation theorem (EAT)~\cite{dupuis2016entropy}. 
		The EAT gives a way of bounding conditional smooth min- and max-entropies in sequential processes in which certain systems of interest are being produced one after the other in overall $n$ steps (not necessarily in an IID way). It, roughly, states that the total amount of smooth entropy accumulated during the \emph{entire} process is $n$ times the von Neumann entropy produced in a \emph{single} step of the process. In our context, this translates to saying that the total amount of smooth max-entropy can be related to the von Neumann entropy considered in Section~\ref{sec:single_round_ent}.
		All the definitions and statements of the EAT  which are necessary for our work are presented in Section~\ref{sec:eat_def_app}.
	
		\subsubsection{Prerequisites of the EAT}\label{sec:prereq_eat}
	
			Before applying the EAT, we show that the prerequisites of the EAT hold.

			Specifically, $\mathcal{N}_i:R_{i-1}\rightarrow R_i\hat{A}_i\hat{B}_iA_iB_iC_iD_iX_iY_iW_i$ are the channels defined by Protocol~\ref{pro:mod_ent_test} and will act as our ``EAT channels''.
			We need to show that these are indeed EAT channels, i.e., that they fulfil Definition~\ref{def:eat_channels} (see Section~\ref{sec:eat_def_app}). To verify that this is the case note the following:
			\begin{enumerate}
				\item $W_i$ are classical finite-dimensional systems. $ d_{\hat{A}_i} = 2$ (due to the projection step of Protocol~\ref{pro:mod_ent_test}).
				\item The classical value $W_i$ is a function of $A_iB_iX_iY_i$. Hence, it can be measured from the output of the EAT channels (for any input state) without modifying the state.
				\item The necessary Markov-chain conditions hold, as stated in the next lemma.
			\end{enumerate}
	
			\begin{lemma}\label{lem:markov_cond}
				For all $i\in[n]$ and any initial state,
				\[
					\hat{A}_{1}^{i-1} \leftrightarrow  \hat{B}_{1}^{i-1}A_{1}^{i-1}B_{1}^{i-1}C_{1}^{i-1}D_{1}^{i-1}X_{1}^{i-1}Y_{1}^{i-1} E  \leftrightarrow \hat{B}_iA_iB_iC_iD_iX_iY_i 
				\]
				holds for the final state  $\tau$ of Protocol~\ref{pro:mod_ent_test}.
			\end{lemma}
			The crucial ingredient in the proof is Lemma~\ref{lem:symm_dec}, which asserts that the twirling step (Step~\ref{prostep:symm_maps} of Protocol~\ref{pro:mod_ent_test}) decouples the states kept in each round from the other registers; see Appendix~\ref{sec:proofs_smo_ent} for the details.
			It now becomes clear why the twirling step is necessary -- without it the required Markov-chain conditions do not hold. 
			For example, one can imagine a source that creates two bipartite states~$\phi^1$ and~$\phi^2$  entangled \emph{with one another}. In such a case the above Markov-chain conditions do not hold. 
			Step~\ref{prostep:symm_maps} therefore \emph{enforces} the necessary conditions while not destroying the entanglement between $\hat{A}_i$ and $\hat{B}_i$.

		\subsubsection{Max-tradeoff function}\label{sec:max_tradeoff_func}
	
			To apply the EAT we need to define a concave max-tradeoff function, as defined in Definition~\ref{def:max_tradeoff_func}. We construct one by following similar steps used to define min-tradeoff functions in~\cite{arnon2016simple}.
				
			Let $p$ be a probability distribution over $\{0,1,\perp\}$ resulting from the observed data $w$, i.e., $p(\tilde{w})=| \left\{ i | w_i = \tilde{w} \right\} |/n$ for $\tilde{w}\in\{0,1,\perp\}$,
			and define 
			\begin{equation*}
				\Sigma^p = \{ \sigma | \mathcal{N}_i(\sigma)_{W_i} = p \} \;.
			\end{equation*}
			
			Below we focus on probability distributions for which $p(0)+p(1)=\gamma$. The reason is that for $p(0)+p(1)\neq \gamma$ the set $\Sigma^p$ is empty and the condition on the max-tradeoff function becomes trivial. 
			For such $p$ we can write $\omega=\frac{p(1)}{p(0)+p(1)}=\frac{p(1)}{\gamma}$.
			
			Lemma~\ref{lem:final_single_round_ent} can now be used to define a max-tradeoff function for any $p$ with $\frac{p(1)}{\gamma} \in\left[\frac{3}{4},\frac{2+\sqrt{2}}{4}\right]$. 
			Define a function $f$ by  
			\begin{equation*}\label{eq:def_f}
			f(p) \,=\,  \begin{cases}
					\left(1-\gamma \right) \cdot g\left(\frac{p(1)}{\gamma} \right) &  \frac{p(1)}{\gamma}\in\left[0,\frac{2+\sqrt{2}}{4}\right] \\
					\gamma-1& \frac{p(1)}{\gamma}\in\left[\frac{2+\sqrt{2}}{4},1\right] \;,
					\end{cases}
			\end{equation*}
			where $g$ is as in Lemma~\ref{lem:final_single_round_ent}:
			\[
				g(\omega)=2h\left(\frac{1}{2}-\frac{2\omega -1}{\sqrt{2}}\right) -1
			\]
			for $h$ the binary entropy function.

			From Lemma~\ref{lem:final_single_round_ent} we have 
			\begin{equation*}\label{eq:one_box_entropy_final}
				\sup_{\sigma \in \Sigma^p} H(\hat{A}_i|\hat{B}_iA_iB_iC_iD_iX_iY_i)_{\mathcal{N}_i(\sigma)} \leq \left(1-\gamma \right) \cdot g\left(\frac{p(1)}{\gamma} \right) \;,
			\end{equation*}
			and, hence, it follows that any choice of $f_{\max}(p)$ that is differentiable and satisfies $f_{\max}(p) \geq f(p)$ for all $p$ will be a valid max-tradeoff function for our EAT channels.
			
			For the final bound derived using the EAT to be meaningful we choose $f_{\max}$ for which $\| \triangledown f_{\max} \|_{\infty}$ is finite (for $f$ given above the derivative at $\frac{p(1)}{\gamma}=\frac{2+\sqrt{2}}{4}$ is infinite).
			Let 
			\begin{equation*}\label{eq:f_max_choice}
				f_{\max}\left(p,p_t\right) = \begin{cases}
					f\left(p\right) &  p(1) \leq p_t(1) \\
					a(p_t) \cdot p(1) + b(p_t) & p(1)> p_t(1)\;,
				\end{cases}
			\end{equation*}
			where $p_t$ is a probability distribution over $\{0,1,\perp\}$ and
			\begin{equation*}\label{eq:derivative} 
				a(p_t)= \frac{\mathrm{d}}{\mathrm{d}p(1)} f(p)\big|_{p_t}  \qquad\text{and} \qquad b(p_t) =  f(p_t)-a(p_t)\cdot p_t(1). 
			\end{equation*}
			
			It follows from the definition of $a$ and $b$ given in the above equation that $f_{\max}$ is differentiable and, for any~$p_t$, $\| \triangledown f_{\max}(\cdot,p_t) \|_{\infty} \leq a(p_t)$. 
			Furthermore, as $f$ is a concave function, we also have that $f_{\max}(p) \geq f(p)$ for all $p$. Thus, $f_{\max}(p)$ is indeed a max-tradeoff function.
	
		\subsubsection{Applying the EAT}
		
			We are finally ready to apply the EAT, stated as Theorem~\ref{thm:eat}, to derive an upper-bound on the smooth max-entropy. 
			The smooth max-entropy rate is governed by the following functions:
			\begin{align}
				&f(p) =  
					\begin{cases}
						 \left(1-\gamma \right) \cdot g\left(\frac{p(1)}{\gamma} \right) &  \frac{p(1)}{\gamma}\in\left[0,\frac{2+\sqrt{2}}{4}\right] \\
						\gamma-1& \frac{p(1)}{\gamma}\in\left[\frac{2+\sqrt{2}}{4},1\right] \;,
					\end{cases}\notag\\
				&f_{\max}\left(p,p_t\right) = 
					\begin{cases}
						f\left(p\right)&  p(1) \leq p_t(1) \;  \\
						\frac{\mathrm{d}}{\mathrm{d}p(1)} f(p)\big|_{p_t}  \cdot p(1)+ \Big( f(p_t) -	\frac{\mathrm{d}}{\mathrm{d}p(1)} f(p)\big|_{p_t} \cdot p_t(1) \Big)& p(1)> p_t(1)\;,
					\end{cases} \nonumber \\
				&\eta(p,p_t,\varepsilon_{\mathrm{smo}},\varepsilon_{\mathrm{snd}}) =  
					f_{\max}\left(p, p_t\right) + \frac{1}{\sqrt{n}}2\left( \log 5 + \Big| \frac{\mathrm{d}}{\mathrm{d}p(1)} g(p)\big|_{p_t}  \Big| \right)\sqrt{1-2 \log (\varepsilon_{\mathrm{smo}} \cdot 	\varepsilon_{\mathrm{snd}})}\;, \nonumber\\
				&\eta_{\mathrm{opt}}(\varepsilon_{\mathrm{smo}}, \varepsilon_{\mathrm{snd}}) =
					 \min_{p_t: \frac{3}{4} < \frac{p_t(1)}{\gamma} < \frac{2+\sqrt{2}}{4}} \; \eta(\omega_{\mathrm{exp}}\gamma - \delta_{\mathrm{est}},p_t,\varepsilon_{\mathrm{smo}},\varepsilon_{\mathrm{snd}})\;. \label{eq:eta_opt}
			\end{align}

			\begin{lemma}\label{lem:ent_rate}
				For any source and measurement device in the considered setting, let $\tau$ the state generated using Protocol~\ref{pro:mod_ent_test}, $\Omega$ the event that Protocol~\ref{pro:mod_ent_test} does not abort, and $\tau_{|\Omega}$ the state conditioned on $\Omega$.
				Then, for any $\varepsilon_{\mathrm{smo}},\varepsilon_{\mathrm{snd}}\in (0,1)$, either Protocol~\ref{pro:mod_ent_test} aborts with probability greater than $1-\varepsilon_{\mathrm{snd}}$ or
				\begin{equation*}
					 H^{\varepsilon_{\mathrm{smo}}}_{\max} \left( \hat{A}|\hat{B}ABCDXY \right)_{\tau_{|\Omega}} < n\cdot \eta_{\mathrm{opt}}(\varepsilon_{\mathrm{smo}},\varepsilon_{\mathrm{snd}}) \;,
				\end{equation*}
				where  $\eta_{\mathrm{opt}}$ is defined in Equation~\eqref{eq:eta_opt}. 
			\end{lemma}
			
			\begin{proof}
				In Sections~\ref{sec:prereq_eat} and~\ref{sec:max_tradeoff_func} we showed that the prerequisites of the EAT hold and that $f_{\max}$ defined above is a max-tradeoff function for the EAT channels defined via Protocol~\ref{pro:mod_ent_test}. The statement of the lemma then follows by applying the EAT, stated as Theorem~\ref{thm:eat}, with the above choices.
			\end{proof}

	
	\subsection{Final statement}

		We are ready to prove our main theorem, Theorem~\ref{thm:main_thm}. For convenience, we restate the theorem here. 
		$\mathcal{S}^{\mathrm{honest}}$ and $\mathcal{D}^{\mathrm{honest}}$ are as discussed in Section~\ref{sec:diec_result} (or see Section~\ref{sec:comp_pro} below).
		
		\begin{customthm}{\ref{thm:main_thm}}[Main theorem]
			For any $n\in \mathbb{N_+} $, $\varepsilon_{\mathrm{dist}},\varepsilon_{\mathrm{snd}}\in [0,1]$, and $\varepsilon_{\mathrm{smo}}\in[0,\sqrt{\varepsilon_{\mathrm{dist}}})$, Protocol~\ref{pro:ent_test} is a DIEC protocol with:
			\begin{enumerate}
				\item Noise-tolerance (completeness): The probability that $\mathrm{P}$ aborts when applied on any $\phi^{\mathrm{honest}}\in\mathcal{S}^{\mathrm{honest}}$ using honest measurement devices from $\mathcal{D}^{\mathrm{honest}}$  is at most $\varepsilon_{\mathrm{cmp}} \leq \exp(- 2 n\delta_{\mathrm{est}}^2 )$.
				\item Entanglement certification (soundness): For \emph{any} source and measurement devices in the considered setting, either Protocol~\ref{pro:ent_test} aborts with probability greater than $1-\varepsilon_{\mathrm{snd}}$ when applied on $\phi$ or 
				$E_D^{n,\varepsilon_{\mathrm{dist}}}(\rho_{|\Omega}) \geq r$ for $r=\log(L)/ n$ and
				\[
					\log L= - n\cdot \eta_{\mathrm{opt}}(\varepsilon_{\mathrm{smo}},\varepsilon_{\mathrm{snd}}) -4\log\left(  \frac{1}{\sqrt{\varepsilon_{\mathrm{dist}}}-\varepsilon_{\mathrm{smo}}}\right)\;,
				\]
				where $\eta_{\mathrm{opt}}$ is defined in Equation~\eqref{eq:eta_opt}. 
			\end{enumerate}
		\end{customthm}
		
		The theorem follows from the combination of Lemmas~\ref{lem:completeness} and~\ref{lem:soundness} given below.
		
		\subsubsection{Noise-tolerance (completeness)}\label{sec:comp_pro}
			
			As the honest source and devices we choose to consider a source that produces identical and independent copies of a state $\phi^i=\sigma$ and  measurement devices that apply the same measurements in each round when they are used. The state $\sigma$ and the measurements are such that the winning probability achieved in the CHSH game is at least $\omega_{\mathrm{exp}}$. 
			For example, one can choose $\mathcal{D}^{\mathrm{honest}}$ to include the measurement devices that apply the optimal measurements performed in the CHSH game and the
			 set $\mathcal{S}^{\mathrm{honest}}$ to include all states $\phi^{\mathrm{honest}}=\sigma^{\otimes n}$ for~$\sigma$ any noisy maximally entangled state that will result in winning probability $\geq \omega_{\mathrm{exp}}$.
			The following lemma bounds the probability of Protocol~\ref{pro:ent_test} aborting when using an honest device as above.
	
			\begin{lemma}\label{lem:completeness}
				The probability that Protocol~\ref{pro:ent_test} aborts for an honest implementation discussed above is at most $\varepsilon_{\mathrm{cmp}} \leq \exp(- 2 n\delta_{\mathrm{est}}^2 ) $.		
			\end{lemma}
			
			\begin{proof}
				The protocol aborts in Step~\ref{prostep:abort_ET_1} when the sum of the $W_i$ obtained during the test rounds is not sufficiently high (this happens when the estimated Bell violation is too low or when not enough test rounds were chosen). In the honest implementation the products $\chi(T_i=1)\,W_i$ are IID\@ RVs with $\mathbb{E}\left[ \chi(T_i=1)\,W_i\right] = \omega_{\mathrm{exp}}\gamma$. Therefore, we can use Hoeffding's inequality:
				\[
					\varepsilon_{\mathrm{cmp}} =  \Pr \left[ \sum_i \chi(T_i=1)\, W_i \leq \left(\omega_{\mathrm{exp}}\gamma - \delta_{\mathrm{est}}\right) \cdot n\ \right] \leq \exp(- 2 n \delta_{\mathrm{est}}^2 ) \;. \qedhere
				\]
			\end{proof}

		\subsubsection{Entanglement certification (soundness)}
		
			\begin{lemma}\label{lem:soundness}
				For any source and measurement devices in the considered setting, denote by $\rho_{|\Omega}$ the state in the end of the DIEC protocol, Protocol~\ref{pro:ent_test}, conditioned on not aborting.
				 Let $\varepsilon_{\mathrm{dist}},\varepsilon_{\mathrm{snd}}\in[0,1]$ and $\varepsilon_{\mathrm{smo}}\in[0,\sqrt{\varepsilon_{\mathrm{dist}}})$.
				
				Then, either Protocol~\ref{pro:ent_test} aborts with probability greater than $1-\varepsilon_{\mathrm{snd}}$ 
				or $E_D^{n,\varepsilon_{\mathrm{dist}}}(\rho_{|\Omega}) \geq r$ for $r=\log(L)/ n$ and
				\begin{equation}\label{eq:fin_ent_rank}
					\log L= - n\cdot \eta_{\mathrm{opt}}(\varepsilon_{\mathrm{smo}},\varepsilon_{\mathrm{snd}}) -4\log\left(  \frac{1}{\sqrt{\varepsilon_{\mathrm{dist}}}-\varepsilon_{\mathrm{smo}}}\right)\;,
				\end{equation}
				where $\eta_{\mathrm{opt}}$ is defined in Equation~\eqref{eq:eta_opt}. 
			\end{lemma}
			
			\begin{proof}
				Given a source and measurement devices in the considered setting, we first consider the hypothetical scenario in which \emph{Protocol~\ref{pro:mod_ent_test}} is being ran using the given devices.  
				Putting Lemmas~\ref{lem:ent_dist} and~\ref{lem:ent_rate} together we learn that
				either Protocol~\ref{pro:mod_ent_test} aborts with probability greater than $1-\varepsilon_{\mathrm{snd}}$ or
				there exists an entanglement distillation protocol $\Gamma$ such that  
				\begin{equation}\label{eq:fin_dist_tau}
					F\left( \Gamma\left(\tau_{|\Omega}\right), \ket{\Phi^L}\bra{\Phi^L}\right) \leq \varepsilon_{\mathrm{dist}} \;,
				\end{equation}
				for $L$ as in Equation~\eqref{eq:fin_ent_rank}.
				
				According to Lemma~\ref{lem:locc_tau_to_rho}, Equation~\eqref{eq:fin_dist_tau} implies that there exists an entanglement distillation protocol~$\Delta$ such that  
				\[
					F\left( \Delta\left(\rho_{|\Omega}\right), \ket{\Phi^L}\bra{\Phi^L}\right) \leq \varepsilon_{\mathrm{dist}} \;,
				\]
				with the same choice of parameters, i.e., $L$ is as in Equation~\eqref{eq:fin_ent_rank}.
				By the definition of the one-shot distillable entanglement, given in Equation~\eqref{eq:one_shot_E_D}, we get that $E_D^{n,\varepsilon_{\mathrm{dist}}}(\rho_{|\Omega}) \geq r$ for $r=\log(L)/ n$.
		
				Finally, Lemma~\ref{lem:mod_same_stat} tells us that if  Protocol~\ref{pro:mod_ent_test} aborts with probability greater than $1-\varepsilon_{\mathrm{snd}}$ then Protocol~\ref{pro:ent_test} aborts with probability greater than $1-\varepsilon_{\mathrm{snd}}$ as well.
							
				Combining the above observations, the lemma follows.
			\end{proof}
	
			The resulting distillable entanglement rates, $(\log L)/n$, are plotted in Figure~\ref{fig:dist_ent_rate} as a function of the expected winning probability in the CHSH game $\omega_{\mathrm{exp}}$ for different values of $n$. As seen from the figure, as the number of rounds of the protocol $n$ increases our rate approach the optimal rate, for our proof technique, given by the IID asymptotic rate; see Section~\ref{sec:open_quest} for further details.

			We remark that one can also derive improved rates for finite number of rounds $n$ by considering a slightly modified version of our DIEC protocol, similarly to what was done in~\cite[Appendix~B]{arnon2016simple}. As the modification of the protocol, the analysis, and the resulting rates follow directly by combining the analysis done here and that of~\cite[Appendix~B]{arnon2016simple} we do not present the details here.

\section{Open questions}\label{sec:open_quest}
	
	\subsubsection*{Tightness of our result}
	
		As mention in the previous sections, if one chooses to take the path of bounding the one-shot distillable entanglement using the smooth max-entropy, as done in the current work, then our quantitive results are tight to first order of $n$. 
		However, this may not be the only way to go. Considering other proof techniques is crucial in order to achieve a result in which there is positive certified distillable entanglement regime whenever a violation of the CHSH inequality is being detected.
		
		In general, it is known that there are Bell inequalities which can be violated by bound entangled states, i.e., entangled state which cannot be distilled~\cite{vertesi2014disproving} and, thus, for some Bell inequalities a zero rate regime, similar to the one observed here, is of a fundamental nature. 
		However, for the CHSH inequality this is not the case~\cite{masanes2006asymptotic}: bound entangled states cannot violate the CHSH inequality. Hence, asymptotically, one should be able to certify distillable entanglement for any violation.

		One way of assessing how far our results are from the optimal results, achievable using any proof technique, is to find \emph{upper bounds} on the asymptotic distillable entanglement of the states achieving the highest conditional entropy given their Bell violation (see Section~\ref{sec:single_round_ent}).
		One possible starting point is to consider the sets of states described in~\cite{leditzky2017useful}. 
		
	\subsubsection*{Possible extensions of our result}
	
		There are many possible ways of extending our work.
	
		\begin{enumerate}
			\item Our protocol and proof technique can also be modified to work with other Bell inequalities instead of the CHSH. This can potentially increase the rates when considering different types of honest sources of entanglement. For example, if one is interested in a source that emits partially entangled states then it probably makes more sense to consider the tilted CHSH inequalities~\cite{acin2012randomness} rather than the CHSH. The only part of the proof which requiers a modification is the upper bound on the von Neumann entropy for a single round given in Section~\ref{sec:single_round_ent}. We remark that one can achieve such a bound for any Bell inequality for which a robust self-testing result is known, e.g.,~\cite{bamps2015sum}, combined with the continuity of the von Neumann entropy~\cite{winter2016tight}. However, it is likely that such an approach will lead to relatively weak quantitive results. Thus, considering the von Neumann entropy directly for other Bell inequalities is a more promising direction. 
		
			\item An important direction to consider is the extension of our work to entanglement shared between more than two parties. This can then be used, for example, to consider scenarios and results as those derived in~\cite{bancal2014device,mccutcheon2016experimental} and extend them beyond the IID setting.
			
			\item Another possible extension of the analysis done here is to consider DIEC protocols that employ the more general (but less fundamental) separability preserving operations~\cite{rains1997entanglement,cirac2001entangling} rather than LOCC. To do so one should first consider one-shot distillation protocols which use separability preserving operations~\cite{brandao2011one}.

			\item One can also try to bound other operational measures of  entanglement. 
			For example, it will be interesting to lower bound the \emph{one-shot entanglement cost}.\footnote{In a related work~\cite{arnon2017noise}, a lower bound on the \emph{entanglement of formation} of a quantum state (closely related to its entanglement cost) as a function of its Bell violation is derived. The considered setting and type of statement are different than the ones presented here and are not operational in our sense. For further details see~\cite{arnon2017noise}.} Such a bound can be viewed as a ``dual'' of the one achieved in this work.  
			Given the results of~\cite{buscemi2011entanglement}, it is plausible that a similar proof technique as presented here can be used to achieve such a bound.
		
			\item In a different direction, it can also be of interest to consider other settings than the one considered in the current work (as described in Section~\ref{sec:setting}). As different experiments may require different sets of assumptions, formulating other interesting scenarios and modifying the proof  accordingly can be relevant. 
			
			\item Similarly, one may consider device-\emph{dependent} and semi-DI versions of our work. For example, it is possible to study a one-sided DI scenario in which one of the measurement devices is completely characterised. 
			The only part of our proof that needs be to modified in such a case is that given in Section~\ref{sec:single_round_ent} while replacing Bell inequalities with Steering inequalities~\cite{cavalcanti2016quantum}. The rest of the proof will follow as is.
			The additional assumptions can potentially result in certification rates higher than the ones presented in the current work. 
		\end{enumerate}

\section*{Acknowledgments}
We thank Renato Renner, David Sutter, and Thomas Vidick for helpful discussions.
We also thank Valerio Scarani for inviting RAF to visit his group at CQT Singapore, where the work on this project was initiated. 
RAF is supported by the Swiss National Science Foundation (grant No. 200020-135048) via the National Centre of Competence in Research ``Quantum Science and Technology'' and by the US Air Force Office of Scientific Research (grant No. FA9550-16-1-0245).	 
JDB  acknowledges  support  from the Swiss National Science Foundation (SNSF), through the NCCR QSIT and the Grant number PP00P2-150579.	

\appendix


\section{Proofs of the lemmas in Section~\ref{sec:modified protocol}}\label{sec:proofs_mod_prot}

\begin{customlemma}{\ref{lem:proj_idn_subspace}}
	Consider a scenario in which the projection to the two qubit space is applied on the state~$\phi^i$ directly after it is produced by the source in the $i$'th round (i.e., before choosing the value of $T_i$). Denote the resulting state in the end of the $i$'th round in such a case by $\bar{\bar{\rho}}^i$. Then 
	\[
		\bar{\bar{\rho}}^i_{\hat{A}_i\hat{B}_i A_iB_i C_i D_iX_iY_i} = \bar{\rho}^i_{\hat{A}_i\hat{B}_i A_iB_i C_i D_iX_iY_i}  \;,
	\]
	where $\bar{\rho}^i$ is as defined in Equation~\eqref{eq:state_after_proj}.
\end{customlemma}
\begin{proof}
	Firstly, for the rounds in which $T_i=0$ there is clearly no difference between $\bar{\bar{\rho}}^i_{|T_i=0}$ and $\bar{\rho}^i_{|T_i=0}$.
	We show that the same holds also when $T_i=1$. To this end we need to prove that  
	\[
		\bar{\bar{\rho}}^i_{A_iB_iX_iY_i|T_i=1} = \bar{\rho}^i_{A_iB_iX_iY_i|T_i=1}\;.
	\]
	To see that this is indeed the case note that the successive application of the projections given in Equation~\eqref{eq:jordan_proj} and of the measurement as applied to create $\bar{\bar{\rho}}^i$ has the exact same effect as applying the measurement alone.
	Indeed,
	\begin{align*}
		\Pi_{a|x}^{A_i} = \sum_{c_i} \Pi_{a|x}^{A_i} \Pi^{A_i}_{c_i}  
	\end{align*}
	and similarly for Bob. Therefore, after tracing out the block registers $C_i$ and $D_i$, both final states are identical.
	
	As the probability for choosing $T_i=0$ is, obviously, independent of when the projection is made the combination of the above statements implies the lemma.
\end{proof}

\begin{customlemma}{\ref{lem:mod_same_stat}}
	The observed statistics and, hence, the probabilities of aborting Protocol~\ref{pro:ent_test} and Protocol~\ref{pro:mod_ent_test} are the same. That is,
	$\rho_{ABXYTW} = \tau_{ABXYTW}$ and
	$\Pr[\neg\Omega]_{\rho} = \Pr[\neg\Omega]_{\tau}$.
\end{customlemma}
\begin{proof}
	The only difference between Protocol~\ref{pro:ent_test} and Protocol~\ref{pro:mod_ent_test} is in the rounds in which $T_i=0$. The observed statistics over $ABXYT$ and, hence, also $W$, depend however only on the rounds in which $T_i=1$. Thus, $\rho_{ABXYTW} = \tau_{ABXYTW}$ . The event $\Omega$ is defined according to the registers $W$ and therefore the lemma follows.
\end{proof}

\begin{customlemma}{\ref{lem:symm_dec}}
	Let $\tau$ denote the state after all rounds of Protocol~\ref{pro:mod_ent_test} (before conditioning on $\Omega$).
	For every~$i\in[n]$ let $K$ include all registers different than $\hat{A}_i \hat{B}_iC_iD_iT_i$ such that we can write 
	\[
		\tau_{\hat{A}\hat{B}ABCDXYTW} = \tau_{\hat{A}_i \hat{B}_iC_iD_iT_iK} \;.
	\] 
	Then,
	\begin{equation*}
		H(\hat{A}_i|\hat{B}_iC_iD_iT_iK) = H(\hat{A}_i|\hat{B}_iC_iD_iT_i) \;.
	\end{equation*}
\end{customlemma}

\begin{proof}
	We prove below that for all $c_i$ and $d_i$ we have
	\begin{equation}\label{eq:symm_lem_cond_cd}
		H(\hat{A}_i|\hat{B}_iT_iK,C_i=c_i,D_i=d_i) = H(\hat{A}_i|\hat{B}_iT_i,C_i=c_i,D_i=d_i) \;.
	\end{equation}
	Then, using the definition of the conditional entropy we write
	\begin{align*}
		H(\hat{A}_i|\hat{B}_iC_iD_iT_iK) &= \sum_{c_i,d_i} p(c_i,d_i) H(\hat{A}_i|\hat{B}_iT_iK,C_i=c_i,D_i=d_i) \\
		&= \sum_{c_i,d_i} p(c_i,d_i)H(\hat{A}_i|\hat{B}_iT_i,C_i=c_i,D_i=d_i) \\
		&= H(\hat{A}_i|\hat{B}_iC_iD_iT_i)
	\end{align*}
	and the lemma follows.
	We are therefore left to prove that Equation~\eqref{eq:symm_lem_cond_cd} holds. 			

	We split the proof to two parts corresponding to the different values of $T_i$.
	When $T_i=1$ the register~$\hat{A}_i$ has a deterministic value. Hence,
	\begin{equation}\label{eq:ent_equiv_test}
		H(\hat{A}_i|\hat{B}_iK,C_i=c_i,D_i=d_i,T_i=1) = H(\hat{A}_i|\hat{B}_i,C_i=c_i,D_i=d_i,T_i=1) = 0 \;.
	\end{equation}
	
	The interesting case is thus $T_i=0$. 
	For $T_i=0$ we can use the Bell diagonal structure of the state as given in Equation~\eqref{eq:bell_mixt}. Consider the purification of the state $\tau_{\hat{A}_i\hat{B}_i|c_i,d_i}$:
	\begin{equation*}
	\begin{split}
		\ket{\tau}_{\hat{A}_i\hat{B}_iF_i}^{c_i,d_i} =& 
		\sqrt{\lambda_{\Phi^+}}\ket{\Phi^+}_{\hat{A}_i\hat{B}_i}\ket{1}_{F_i}	+
		\sqrt{\lambda_{\Phi^-}}\ket{\Phi^-}_{\hat{A}_i\hat{B}_i}\ket{2}_{F_i}\\
		&+ \sqrt{\lambda_{\Psi^+}}\ket{\Psi^+}_{\hat{A}_i\hat{B}_i}\ket{3}_{F_i}	+
		\sqrt{\lambda_{\Psi^-}}\ket{\Psi^-}_{\hat{A}_i\hat{B}_i}\ket{4}_{F_i} \;.
	\end{split}
	\end{equation*}
	
	Note that $K$ does not include the information encoded in $F_i$ by the definition of $K$.
	Thus, including the register $F_i$ we must have 
	\[
		\tau_{\hat{A}_i \hat{B}_iF_iK|C_i=c_i,D_i=d_i,T_i=0} = \tau_{\hat{A}_i \hat{B}_iF_i|T_i=0} \otimes  \tau_{K|C_i=c_i,D_i=d_i,T_i=0} \;.
	\]
	Moreover, we can freely trace $F_i$ out while preserving the tensor product structure
	\[
		\tau_{\hat{A}_i \hat{B}_iK|C_i=c_i,D_i=d_i,T_i=0} = \tau_{\hat{A}_i \hat{B}_i|T_i=0} \otimes  \tau_{K|C_i=c_i,D_i=d_i,T_i=0} \;,
	\]
	from which it follows that 
	\begin{equation}\label{eq:ent_equiv_gen}
		H(\hat{A}_i|\hat{B}_iK,C_i=c_i,D_i=d_i,T_i=0) = H(\hat{A}_i|\hat{B}_i,C_i=c_i,D_i=d_i,T_i=0) \;.
	\end{equation}
	The combination of Equations~\eqref{eq:ent_equiv_test} and~\eqref{eq:ent_equiv_gen}, together with the definition of the conditional entropy, implies the lemma. \qedhere
\end{proof}

\begin{customlemma}{\ref{lem:locc_tau_to_rho}}
	Let $\Gamma$ denote the LOCC protocol used to distill 
	$\ket{\Phi^L}= \frac{1}{\sqrt{L}}\sum_{i=1}^L \ket{i}\ket{i}$
	from $\tau_{|\Omega}$ with error probability $\varepsilon$. 
	Then, there exists another LOCC protocol $\Delta$ which can be used to distill $\ket{\Phi^L}$ from $\rho_{|\Omega}$ with the same error probability.	
\end{customlemma}
\begin{proof}
	The only difference between Protocol~\ref{pro:ent_test} and Protocol~\ref{pro:mod_ent_test} is the addition of Steps~\ref{prostep:state_after_proj} and~\ref{prostep:symm_maps}.
	
	These steps are such that their effect is restricted to the round in which they are applied. That is, whether we apply them or not in round $i$ does not effect all other rounds $j\neq i$. In other words, they commute with the rest of the operations made in the protocol. We can therefore postpone them (for all steps $i$ with $T_i=0$) to the end of the protocol. 
	
	Furthermore, according to Lemma~\ref{lem:mod_same_stat} the additional steps do not change the observed statistics and the probability of the event $\Omega$. Thus, we can also postpone them for after making the projection on $\Omega$. 
	
	Denoting the combination of all the projections and rotations for all relevant rounds by the map~$\Lambda$, the above means that the relation
	\begin{equation}\label{eq:rho_to_tau_LOCC}
		\Lambda\left(\rho_{|\Omega}\right) = \tau_{|\Omega} 
	\end{equation}
	holds. Moreover, $\Lambda$ can be implemented using only LOCC by definition of the Steps~\ref{prostep:state_after_proj} and~\ref{prostep:symm_maps}. 
	
	Let $\Gamma$ denote the LOCC protocol used to distill $\ket{\Phi^L}$ from $\tau_{|\Omega}$ with error probability $\varepsilon$. We define
	\begin{equation}\label{eq:concat_locc}
		\Delta = \Gamma \cdot \Lambda
	\end{equation}
	to be the successive application of $\Lambda$ and $\Gamma$. As both are LOCC protocols, $\Delta$ by itself is an LOCC protocol. 
	From Equations~\eqref{eq:rho_to_tau_LOCC} and~\eqref{eq:concat_locc} it follows that 
	\[
		\Delta\left(\rho_{|\Omega}\right) = \Gamma \left( \Lambda \left(\rho_{|\Omega}\right) \right) = \Gamma \left( \tau_{|\Omega} \right) \;.
	\]
	Additionally, the operation $\Lambda$ alway succeeds (i.e., Equation~\eqref{eq:rho_to_tau_LOCC} is always true). 
	Therefore, $\Delta$ is an LOCC protocol  which can be used to distill $\ket{\Phi^L}$ from $\rho_{|\Omega}$ with the same error probability as~$\Gamma$.
\end{proof}

\section{Proofs of the lemmas in Section~\ref{sec:single_round_ent}}\label{sec:proofs_sing_ent}

\begin{customlemma}{\ref{lem:single_ent_diagonal_Bell}}
	For any Bell diagonal state $\sigma_{\hat{A}_i\hat{B}_i}$ as in Equation~\eqref{eq:bell_mixt}  that can be used to violate the CHSH inequality with violation $\beta\in\left[2,2\sqrt{2}\right]$,
	\begin{equation*}
		H(\hat{A}_i|\hat{B}_i)_{\sigma} \leq \mathscr{H}(\beta) - 1 \;,
	\end{equation*}
	where
	\begin{equation}\label{eq:H_of_beta_apnd}
		\mathscr{H}(\beta) = 2h\left(\frac{1}{2}-\frac{\beta}{4\sqrt{2}}\right) 
	\end{equation}
	and $h$ is the binary entropy function.
\end{customlemma}
\begin{proof}
	First note that 
	\begin{equation}\label{eq:cond_ent_to_comb}
		H(\hat{A}_i|\hat{B}_i)  = H(\hat{A}_i\hat{B}_i) - H(\hat{B}_i) = H(\hat{A}_i\hat{B}_i) -1 
	\end{equation}
	since the marginal $\sigma_{\hat{B}_i}$ is a completely mixed qubit state. 
	
	We are left to upper bound $H(\hat{A}_i\hat{B}_i)$ which, given that the state is Bell diagonal, is simply the Shannon entropy $H(\vec{\lambda})$ of the probability distribution $\vec{\lambda}=(\lambda_{\Phi^+},\lambda_{\Phi^-},\lambda_{\Psi^+},\lambda_{\Psi^-})$ defined by the four eigenvalues of $\sigma$. Hence, our goal is to maximise $H(\vec{\lambda})$ under the constraint of having the correct Bell violation $\beta$.  
	
	Following~\cite[Lemma 7]{pironio2009device} and~\cite[Equation (21)]{horodecki1995violating} we translate the constraint on the violation of the state to the following constraint on the eigenvalues:
	\begin{equation}\label{eq:bell_val_const}
	\begin{split}
		\beta = \max \Big\{ 
					& 2\sqrt{2} \sqrt{(\lambda_{\Phi^+}-\lambda_{\Psi^+})^2+(\lambda_{\Phi^-}-\lambda_{\Psi^-})^2} \;, \\
					& 2\sqrt{2} \sqrt{(\lambda_{\Phi^+}-\lambda_{\Psi^-})^2+(\lambda_{\Phi^-}-\lambda_{\Psi^+})^2} \;, \\ 
					& 2\sqrt{2} \sqrt{(\lambda_{\Phi^+}-\lambda_{\Phi^-})^2+(\lambda_{\Psi^+}-\lambda_{\Psi^-})^2} 
				\Big\} \;.
	\end{split}
	\end{equation}
	
	To simplify the constraint we observe that both the objective function $H(\vec{\lambda})$ and the above constraint are invariant under the exchange of the eigenvalues with one another. Thus, we can assume without loss of generality that the optimal solution is restricted by, say, the first term in Equation~\eqref{eq:bell_val_const}. That is, the eigenvalues of the state that maximise the entropy are such that 
	\begin{align*}
		&\beta = 2\sqrt{2} \sqrt{(\lambda_{\Phi^+}-\lambda_{\Psi^+})^2+(\lambda_{\Phi^-}-\lambda_{\Psi^-})^2} \\
		&\beta \leq 2\sqrt{2} \sqrt{(\lambda_{\Phi^+}-\lambda_{\Psi^-})^2+(\lambda_{\Phi^-}-\lambda_{\Psi^+})^2} \\
		&\beta \leq 2\sqrt{2} \sqrt{(\lambda_{\Phi^+}-\lambda_{\Phi^-})^2+(\lambda_{\Psi^+}-\lambda_{\Psi^-})^2}. 
	\end{align*}
	To see that this is indeed the case one can assume by contradiction that the second term in Equation~\eqref{eq:bell_val_const}, and not the first one, is the one restricting the optimal solution. Then, by exchanging the values of $\lambda_{\Psi^+}$ with $\lambda_{\Psi^-}$ we get a different state, which is restricted by the first term instead of the second one, but attains the same value $H(\vec{\lambda})$. Hence, the solution defined by this exchange must be an optimal solution as well and we can work with it instead. 
	
	We can therefore restrict our attention to the optimisation problem 
	\begin{equation}\label{eq:optimisation_prob}
	\begin{split}
		\underset{\vec{\lambda}}{\max} \quad &H(\vec{\lambda}) \\
		\text{s.t.} \quad &\beta = 2\sqrt{2} \sqrt{(\lambda_{\Phi^+}-\lambda_{\Psi^+})^2+(\lambda_{\Phi^-}-\lambda_{\Psi^-})^2} \\
		& \lambda_{\Phi^+},\lambda_{\Psi^+},\lambda_{\Phi^-},\lambda_{\Psi^-}\geq 0 \\
		& \lambda_{\Phi^+}+\lambda_{\Psi^+}+\lambda_{\Phi^-}+\lambda_{\Psi^-} =1
	\end{split}
	\end{equation}

	The constraints stated in the optimisation problem given as Equation~\eqref{eq:optimisation_prob} imply the constraint:
	\begin{equation}\label{eq:comb_const}
		\lambda_{\Phi^+} = \frac{1}{2}\left(1- \lambda_{\Phi^-} - \lambda_{\Psi^-} + \sqrt{\frac{\beta^2}{8} - \left(\lambda_{\Phi^-} - \lambda_{\Psi^-}\right)^2}\right) \;.
	\end{equation}
	Hence, $H(\vec{\lambda})$ can be written as a function of only two variables, $\lambda_{\Phi^-}$ and $\lambda_{\Psi^-}$, for any value of~$\beta$.

	Let us identify the region of interest in this plane. First, we have the linear conditions
	\begin{eqnarray}
	\lambda_{\Phi^-},\lambda_{\Psi^-} \geq 0 \label{eq:pos}\\
	\lambda_{\Phi^-} + \lambda_{\Psi^-} \leq 1 \label{eq:sum}
	\end{eqnarray}
	Noticing that the optimisation \eqref{eq:optimisation_prob} is invariant under exchange of $(\lambda_{\Phi^+},\lambda_{\Psi^+})$ and $(\lambda_{\Phi^-},\lambda_{\Psi^-})$, we can restrict our attention to solutions for which $(\lambda_{\Phi^-}-\lambda_{\Psi^-})^2 \leq (\lambda_{\Phi^+}-\lambda_{\Psi^+})^2$. Expressed in terms of our two independent variable $\lambda_{\Phi^-}$ and $\lambda_{\Psi^-}$, this condition is equivalent to
	\begin{equation}\label{eq:condition1}
		|\lambda_{\Phi^-} - \lambda_{\Psi^-}|\leq\frac{\beta}{4}.
	\end{equation}
	This condition defines a diagonal band of interest in the space of $\lambda_{\Phi^-}$ and $\lambda_{\Psi^-}$.
	
	Using the condition $\lambda_{\Phi^+}\geq 0$, we further restrict our region of interest:
	\begin{equation}\label{eq:condition2}
		(\lambda_{\Phi^-}-\frac12)^2 + (\lambda_{\Psi^-}-\frac12)^2 \geq \left(\frac{\beta}{4}\right)^2
	\end{equation}
	This condition describes the outside of a disk centred at $\lambda_{\Phi^-}=\lambda_{\Psi^-}=\frac12$. For all $\beta>2$, this identifies a unique region within the set defined by~\eqref{eq:pos}-\eqref{eq:condition1}. Our objective function, the entropy $H(\vec{\lambda})$ is presented in Figure~\ref{fig:H_lambda} on our region of interest for $\beta=2.5$.
	
	The above optimisation can thus be re-expressed in terms of only two variables as
	\begin{equation}\label{eq:optimisation_prob2}
	\begin{split}
		\underset{\lambda_{\Phi^-},\lambda_{\Psi^-}}{\max} \quad &H(\vec{\lambda}) \\
		\text{s.t.} \quad & \lambda_{\Phi^+} = \frac{1}{2}\left(1- \lambda_{\Phi^-} - \lambda_{\Psi^-} + \sqrt{\frac{\beta^2}{8} - \left(\lambda_{\Phi^-} - \lambda_{\Psi^-}\right)^2}\right) \\
		& \lambda_{\Psi^+} = \frac{1}{2}\left(1- \lambda_{\Phi^-} - \lambda_{\Psi^-} - \sqrt{\frac{\beta^2}{8} - \left(\lambda_{\Phi^-} - \lambda_{\Psi^-}\right)^2}\right) \\
		& \lambda_{\Phi^-},\lambda_{\Psi^-}\geq 0 \\
		& (\lambda_{\Phi^-}-\frac12)^2 + (\lambda_{\Psi^-}-\frac12)^2 \geq \left(\frac{\beta}{4}\right)^2.
	\end{split}
	\end{equation}
	
	One can solve this optimisation problem numerically within this region; the solution is given by  
	\begin{equation}\label{eq:opt_sol}
	\begin{split}
		&\lambda_{\Phi^+}^*=  \left(\frac{1}{2}-\frac{\beta}{4\sqrt{2}}\right)^2 \;; \\
		&\lambda_{\Psi^+}^* = \left(\frac{1}{2}+\frac{\beta}{4\sqrt{2}}\right)^2 \;; \\
		&\lambda_{\Phi^-}^* = \lambda_{\Psi^-}^* = \left(\frac{1}{2}-\frac{\beta}{4\sqrt{2}}\right)\left(\frac{1}{2}+\frac{\beta}{4\sqrt{2}}\right)  \;.
	\end{split}
	\end{equation}
	The white point in Figure~\ref{fig:H_lambda} denotes this solution.
	
	Now that we have a solution in hand, we can verify that it is indeed a correct solution (i.e., there are no numerical errors) which is locally optimal. To do so we need to check that the local gradient vanishes at our suggested solution in order to verify that it is indeed the maxima.
	
	The objective function of~\eqref{eq:optimisation_prob2} is invariant under the exchange $\lambda_{\Phi^-}\leftrightarrow\lambda_{\Psi^-}$. Its derivative in the direction $\lambda_{\Psi^-}-\lambda_{\Phi^-}$ must thus vanish at this point. To conclude about the local optimality of this solution, we are left to show that the derivative, on the curve $\lambda_{\Phi^-}=\lambda_{\Psi^-}$, vanishes at the point given in Equation~\eqref{eq:opt_sol}.
	
	Under the constraint $\lambda_{\Phi^-}=\lambda_{\Psi^-}$, $H(\vec{\lambda})$ can be written as a function of the single variable, say, $\lambda_{\Psi^-}$. Explicitly, it reads
	\begin{equation*}
	\begin{split}
		H(\vec{\lambda}) =  - 2 \lambda_{\Psi^-}\log\left(\lambda_{\Psi^-}\right)
		&-\left(\frac{1}{2}-\lambda_{\Psi^-} + \frac{\beta}{4\sqrt{2}}\right)\log\left(\frac{1}{2}-\lambda_{\Psi^-} + \frac{\beta}{4\sqrt{2}}\right) \\
		&- \left(\frac{1}{2}-\lambda_{\Psi^-} - \frac{\beta}{4\sqrt{2}}\right)\log\left(\frac{1}{2}-\lambda_{\Psi^-} - \frac{\beta}{4\sqrt{2}}\right)\;.
	\end{split}
	\end{equation*}
	
	Thus, 
	\[
		\frac{\mathrm{d}H(\vec{\lambda})}{\mathrm{d}\lambda_{\Psi^-}} = \log\left(\left(\frac{1}{2}-\lambda_{\Psi^-}\right)^2-\frac{\beta}{32}\right) - 2\log\left(\lambda_{\Psi^-}\right)
	\]
	and we have $\frac{\mathrm{d}H(\vec{\lambda})}{\mathrm{d}\lambda_{\Psi^-}} \Big|_{\lambda_{\Psi^-}=\lambda_{\Psi^-}^*} =0$ 
	for $\lambda_{\Psi^-}^*$ as in Equation~\eqref{eq:opt_sol}, as we were set to show.
	
	Therefore, we have identified a local optimum for our objective function. According to the plot given in Figure~\ref{fig:H_lambda}, this local optimum is also a global optimum for Equation~\eqref{eq:optimisation_prob2} and hence Equation~\eqref{eq:optimisation_prob}. As the values of the optimal eigenvalues, given in Equation~\eqref{eq:opt_sol}, depend only on the CHSH violation~$\beta$ we can write $H(\hat{A}_i\hat{B}_i) = H(\vec{\lambda}^*)=\mathscr{H}(\beta)$ as in Equation~\eqref{eq:H_of_beta_apnd}.
	Combined with Equation~\eqref{eq:cond_ent_to_comb} the lemma follows. \qedhere

\end{proof}

\begin{customlemma}{\ref{lem:single_ent_convex_comb_Bell}}
	For any $\beta\in\left[2,2\sqrt{2}\right]$ let $\mathscr{H}(\beta)$ be as in Equation~\eqref{eq:H_of_beta_apnd}. 
	Then, for any state $\sigma_{\hat{A}_i\hat{B}_iC_iD_i}$ as in Equation~\eqref{eq:mixt_of_bell_mix} that can be used to violate the CHSH inequality with violation $\beta$,
	\begin{equation*}
		H(\hat{A}_i|\hat{B}_iC_iD_i)_{\sigma} \leq \mathscr{H}(\beta) - 1 \;.
	\end{equation*}
\end{customlemma}

\begin{proof}
	We use the convex combination structure of $\sigma_{\hat{A}_i\hat{B}_iC_iD_i}$ given in Equation~\eqref{eq:mixt_of_bell_mix}. 
	Let $\beta^{c_id_i}$ be the violation of the conditional state $\sigma_{|c_i,d_i}$.  Then, 
	\begin{align*}
		H(\hat{A}_i|\hat{B}_iC_iD_i) &= \sum_{c_i,d_i} p(c_i,d_i)H(\hat{A}_i|\hat{B}_i,C_i=c_i,D_i=d_i) \\
		&\leq \sum_{c_i,d_i} p(c_i,d_i) \left(H(\beta^{c_id_i}) -1 \right)\\
		&\leq H\left(\sum_{c_i,d_i} p(c_i,d_i)  \beta^{c_id_i}\right) -1 \\
		&=  \mathscr{H}(\beta) -1 \;,
	\end{align*}
	where the first equality follows from the definition of the conditional entropy, the second step holds due to  Lemma~\ref{lem:single_ent_diagonal_Bell}, in the third step we used the concavity of $\mathscr{H}(\beta)$, and the last step follows from $\beta = \sum_{c_i,d_i} p(c_i,d_i) \beta^{c_id_i}$.
\end{proof}

\begin{customlemma}{\ref{lem:final_single_round_ent}}
	For any $\omega_{\mathrm{th}}\in\left[\frac{3}{4},\frac{2+\sqrt{2}}{4}\right]$, let $\Sigma = \{ \sigma | w\left(\sigma\right) \geq \omega_{\mathrm{th}}  \}$ . Then,
	\begin{equation*}
		\sup_{\sigma \in \Sigma} H(\hat{A}_i|\hat{B}_iA_iB_iC_iD_iX_iY_i)_{\mathcal{N}_i(\sigma)} \leq \left(1-\gamma \right) \cdot g(\omega_{\mathrm{th}}) \;,
	\end{equation*}
	where 
	\[
		g(\omega_{\mathrm{th}})=2h\left(\frac{1}{2}-\frac{2\omega_{\mathrm{th}} -1}{\sqrt{2}}\right) -1
	\]
	 and	$h$ is the binary entropy function.
\end{customlemma}
\begin{proof}
	We start with the following sequence of equalities that holds for any $\sigma$:
	\begin{align}
		H(\hat{A}_i|\hat{B}_iA_iB_iC_iD_iX_iY_i)_{\mathcal{N}_i(\sigma)} =& \Pr[A_i,B_i=\perp]\cdot H(\hat{A}_i|\hat{B}_iX_iY_iC_iD_i,A_i=\perp,B_i=\perp)_{\mathcal{N}_i(\sigma)} \nonumber\\
		& + \Pr[A_i,B_i\neq\perp]\cdot H(\hat{A}_i|\hat{B}_iX_iY_iC_iD_i,A_i\neq\perp,B_i\neq\perp)_{\mathcal{N}_i(\sigma)} \label{eq:vN_step1} \\
		=& \Pr[T_i=0]\cdot H(\hat{A}_i|\hat{B}_iA_iB_iC_iD_iX_iY_i,T_i=0)_{\mathcal{N}_i(\sigma)} \nonumber\\
		& + \Pr[T_i=1]\cdot H(\hat{A}_i|\hat{B}_iA_iB_iC_iD_iX_iY_i,T_i=1)_{\mathcal{N}_i(\sigma)} \label{eq:vN_step2} \\ 
		=& (1-\gamma)\cdot H(\hat{A}_i|\hat{B}_iA_iB_iC_iD_iX_iY_i,T_i=0)_{\mathcal{N}_i(\sigma)} \nonumber \\
		& +\gamma \cdot H(\hat{A}_i|\hat{B}_iA_iB_iC_iD_iX_iY_i,T_i=1)_{\mathcal{N}_i(\sigma)} \label{eq:vN_step3}\\
		=&(1-\gamma)\cdot H(\hat{A}_i|\hat{B}_iA_iB_iC_iD_iX_iY_i,T_i=0)_{\mathcal{N}_i(\sigma)} \label{eq:vN_step4} \\
		=&(1-\gamma)\cdot H(\hat{A}_i|\hat{B}_iC_iD_i,T_i=0)_{\mathcal{N}_i(\sigma)} \label{eq:vN_step5}
	\end{align}

	To get Equation~\eqref{eq:vN_step1} we used the definition of the conditional entropy. Equations~\eqref{eq:vN_step2} and~\eqref{eq:vN_step3} follow from the definition of Protocol~\ref{pro:mod_ent_test}. 
	Equation~\eqref{eq:vN_step4} holds since, by definition, the registers $\hat{A}_i\hat{B}_i$ have a deterministic initialisation value when a test is performed (nothing is kept in them) and hence $H(\hat{A}_i|\hat{B}_iA_iB_iC_iD_iX_iY_i,T_i=1)_{\mathcal{N}_i(\sigma)}=0$. 
	Finally, Equation~\eqref{eq:vN_step5} follows from the fact that when $T_i=0$ $A_iB_iX_iY_i$ are all $\perp$.
	
	The lemma then follows by combining Equation~\eqref{eq:vN_step5} and  Lemma~\ref{lem:single_ent_convex_comb_Bell} while using the transformation $\beta = 8\omega_{\mathrm{th}}-4$ to replace $\beta$ with $\omega_{\mathrm{th}}$ in Equation~\eqref{eq:H_of_beta}.
\end{proof}

\section{Proofs of the lemmas in Section~\ref{sec:upper_bound_max_ent}}\label{sec:proofs_smo_ent}

\begin{customlemma}{\ref{lem:markov_cond}}
	For all $i\in[n]$ and any initial state,
	\[
		\hat{A}_{1}^{i-1} \leftrightarrow  \hat{B}_{1}^{i-1}A_{1}^{i-1}B_{1}^{i-1}C_{1}^{i-1}D_{1}^{i-1}X_{1}^{i-1}Y_{1}^{i-1} E  \leftrightarrow \hat{B}_iA_iB_iC_iD_iX_iY_i 
	\]
	holds for the final state  $\tau$ of Protocol~\ref{pro:mod_ent_test}.
\end{customlemma}
\begin{proof}
	The Markov-chain conditions can be written using the mutual information as 
	\[
		I\left(\hat{A}_{1}^{i-1} : \hat{B}_iA_iB_iC_iD_iX_iY_i | \hat{B}_{1}^{i-1}A_{1}^{i-1}B_{1}^{i-1}A_{1}^{i-1}B_{1}^{i-1}X_{1}^{i-1}Y_{1}^{i-1} E\right) =0 \;.
	\]
	
	Using the chain rule for the mutual information the above equation can be alternatively written~as 
	\begin{equation*}
		 \sum_{j\in [i-1]} I\left(\hat{A}_j : \hat{B}_iA_iB_iC_iD_iX_iY_i | \hat{B}_{1}^{i-1}A_{1}^{i-1}B_{1}^{i-1}A_{1}^{i-1}B_{1}^{i-1}X_{1}^{i-1}Y_{1}^{i-1} E \hat{A}_{1}^{j-1} \right) = 0 \;.
	\end{equation*}
	
	We now prove that for all $i\in[n]$ and $j\in [i-1]$
	\begin{equation*}\label{eq:mutual_info_j}
		I\left(\hat{A}_j : \hat{B}_iA_iB_iC_iD_iX_iY_i | \hat{B}_{j}A_{j}B_{j}C_{j}D_{j}X_{j}Y_{j}  \right) = 0 \;,
	\end{equation*}
	which, in turn, implies the lemma.
	
	To this end, we use the following sequence of equations
	\begin{align}\label{eq:ent_ind_for_markov}
		 H\left(\hat{A}_j | \hat{B}_{j}A_{j}B_{j}C_{j}D_{j}X_{j}Y_{j}  \right) 	
		&\leq H\left(\hat{A}_j | \hat{B}_{j} C_{j}D_{j}T_{j}  \right) \nonumber \\
		&=H(\hat{A}_j|\hat{B}_jC_{j}D_{j}T_jK) \nonumber\\
		&\leq H(\hat{A}_j|\hat{B}_jC_{j}D_{j}T_j A_{j}B_{j}X_{j}Y_{j} \hat{B}_iA_iB_iC_iD_iX_iY_i) \;,
	\end{align}
	where $K$ includes all registers of $\tau$ apart from~$\hat{A}_j\hat{B}_jC_{j}D_{j}T_j$ such that we can write
	$\tau_{\hat{A}\hat{B}ABCDXYTW} = \tau_{\hat{A}_j \hat{B}_jC_jD_jT_jK}$. In particular, $K$ includes  
	$A_{j}B_{j}X_{j}Y_{j} \hat{B_i}A_iB_iC_iD_iX_iY_i$. 
	The first and last steps above are due to strong subadditivity of the entropy and the middle one follows from Lemma~\ref{lem:symm_dec}.
	
	Using the definition of the mutual information and Equation~\eqref{eq:ent_ind_for_markov} we get
	\begin{align*}
		I\left(\hat{A}_j : \hat{B}_iA_iB_iC_iD_iX_iY_i | \hat{B}_{j}A_{j}B_{j}C_{j}D_{j}X_{j}Y_{j}  \right) =\;& H\left(\hat{A}_j | \hat{B}_{j}A_{j}B_{j}C_{j}D_{j}X_{j}Y_{j}  \right) \\
		&- H\left(\hat{A}_j | \hat{B}_{j}A_{j}B_{j}C_{j}D_{j}X_{j}Y_{j} \hat{B}_iA_iB_iC_iD_iX_iY_i \right) \\
		\leq\;& H(\hat{A}_j|\hat{B}_jT_j A_{j}B_{j}C_{j}D_{j}X_{j}Y_{j} \hat{B}_iA_iB_iC_iD_iX_iY_i)\\
		&- H\left(\hat{A}_j | \hat{B}_{j}A_{j}B_{j}C_{j}D_{j}X_{j}Y_{j} \hat{B}_iA_iB_iC_iD_iX_iY_i \right) \\
		=\;& 0 
	\end{align*}
	and the lemma follows.
\end{proof}


	
\bibliographystyle{alpha}
\bibliography{refs.bib}

\end{document}